\documentclass[runningheads]{llncs}

\usepackage[T1]{fontenc}
\usepackage[ruled]{algorithm2e} 
\usepackage{latexsym}
\usepackage{amsmath}
\usepackage{amsfonts}
\usepackage{hyperref}
\usepackage{enumerate}
\usepackage{graphicx}
\usepackage{color}

\newcommand{\uvec}[1]{\mathbf{1}_{#1}}

\newcommand{\uo}[1]{U_{\mathrm{O}}^{#1}}
\newcommand{\uoup}[1]{U_{\mathrm{O}\uparrow 1}^{#1}}
\newcommand{\ue}[1]{U_{\mathrm{E}}^{#1}}
\newcommand{\w}[2]{W_{#1}^{#2}}
\newcommand{\vd}[2]{V_{#1}^{#2}}

\renewcommand{\Pr}{\mathbf{P}}
\newcommand{\Ex}{\mathbf{E}}

\newcommand{\conv}{\mathrm{conv}}

\newcommand{\lambdao}{\lambda_{\mathrm{O}}}
\newcommand{\lambdaoup}{\lambda_{\mathrm{O}\uparrow 1}}
\newcommand{\lambdae}{\lambda_{\mathrm{E}}}

\newcommand{\Gammait}{\mathit{\Gamma}}

\SetAlFnt{\small}
\SetAlCapFnt{\small}
\SetAlCapNameFnt{\small}
\SetAlCapHSkip{0pt}
\IncMargin{-\parindent}

\usepackage{color}

\begin{document}

\title{Discrete Two Player All-Pay Auction with Complete Information\thanks{This work was supported by Polish National Science Centre through grant
2018/29/B/ST6/00174.}}

\author{Marcin Dziubi\'{n}ski\inst{1}\orcidID{0000-0003-1756-2424} \and
Krzysztof Jahn\inst{2}}
\authorrunning{M. Dziubi\'{n}ski and K. Jahn}
%
\institute{Institute of Informatics, University of Warsaw, Banacha 2, 02-097 Warsaw, Poland
\email{m.dziubinski@mimuw.edu.pl} \and
Faculty of Mathematics and Information Science, Warsaw University of Technology\\
\email{ka.jahn@student.uw.edu.pl}}
\maketitle              

\begin{abstract}
We study discrete two player all-pay auction with complete information. We provide full characterization of mixed strategy Nash equilibria 
and show that they constitute a subset of Nash equilibria of discrete General Lotto game.
We show that equilibria are not unique in general but they are interchangeable. We also show that equilibrium payoffs are unique, 
unless valuation of at least one of the players is an even integer number. If equilibrium payoffs are not unique, continuum of equilibrium 
payoffs are possible. We compare equilibrium characterization and equilibrium payoffs with the continuous variant of all-pay auctions.

\end{abstract}

\section{Introduction}
\label{sec:intro}

All-pay auction constitutes a fundamental game theoretic model of contests where players exert effort in order to win a prize and only the player exerting 
the most effort wins the prize while other players effort is lost without a reward. This form of strategic interaction underlies economic activities
such as political campaigns~\cite{Snyder89}, R\&D races~\cite{Dasgupta86}, rent seeking and lobbying activities~\cite{Moulin86,HillmanSamet87,HillmanRiley89,BayeKovenockVries93},
competition for a monopoly position~\cite{Ellingsen91}, as well as sport competition~\cite{Szymanski03}.

Full characterisation of equilibria in a variant of all-pay auction where effort of players is continuous was obtained by~\cite{BayeKovenockVries96,HillmanRiley89}.
In many situations, however, it is natural to assume that effort is discrete: monetary expenditure, the time spend on projects, or man-power are usually 
measured in discrete units. This rises a number of questions How does the discrete character of effort affect equilibrium behaviour of auction participants?
Does it benefit the stronger or the weaker side of an auction? How well does equilibrium characterisation based on the continuous model approximates
equilibria in the discrete model? In this paper we address these and similar questions by providing a complete characterisation of mixed strategy Nash equilibria
of discrete all-pay auctions. We show that certain qualitative features of these equilibria are similar to features of equilibria in continuous
all-pay auctions. Let $v_1$ and $v_2$ be valuations of the prize by the two auction participants. Suppose that $v_1 \geq v_2$, so that the second
participant is the weaker one. In equilibrium of the continuous model the weaker player chooses zero effort with probability $1 - v_2/v_1$ and,
with probability $v_2/v_1$ chooses her effort level by mixing uniformly on the interval $[0,v_2]$. The stronger player mixes uniformly on the interval
$[0,v_2]$. We show that in the discrete model the weaker player chooses zero effort with probability close to $1 - v_1/v_2$ and chooses her effort level
by mixing on the interval $[0,v_2]$ with distributions which are convex combinations of distributions which are uniform on even and odd number in 
an interval close to $[0,v_2]$ or are distorted variants of such distributions. Equilibrium payoffs are generically unique. Discreteness of
effort levels benefits the weaker player allowing her to obtain a positive payoff when her prize valuation is close to the prize valuation of the
stronger player. In the case of the stronger player, discreteness of effort levels may be beneficial or not, depending on the prize valuation of the
weaker player.

All-pay auctions with discrete effort levels were studied by~\cite{CohenSela07}. The focus of their paper is the effect of different tie-breaking policies.
In the case of the tie breaking policy where each of the two players wins the prize with probability half in case of a tie, as considered in our paper,
they provide partial characterisation of equilibria in the cases where players valuations are integer numbers.
All-pay auctions are closely related to General Lotto games~\cite{Hart08}. Continuous variant of these games was considered by~\cite{BellCover80,Myerson93,SahuguetPersico06}. In particular, \cite{SahuguetPersico06}	 show that equilibria in these games are related to equilibria in 
continous all pay auctions and exploit this connection to obtain equilibrium characterisation. In this paper we take a reverse approach. We show
that equilibria in discrete all-pay auctions are a subset of equilibria in the corresponding discrete General Lotto games. We then use the full 
characterisation of equilibria in discrete General Lotto games obtained by~\cite{Hart08} and~\cite{Dziubinski12} to obtain full characterisation
of equilibria in discrete all-pay auctions.

The rest of the paper is organized as follows. In Section~\ref{sec:model} we define the model of discrete all-pay auctions. Section~\ref{sec:analysis} 
contains the characterisation results as well as discussion of relation to the continuous variant of all-pay auctions. 
We conclude in Section~\ref{sec:conclusions}. Most part of the proofs is given in the Appendix.

\section{The model}
\label{sec:model}

There are two players, $1$ and $2$, competing for a prize that is worth $v_2 > 0$ for player $2$ and $v_1 \geq v_2$ for player $1$.
Player $2$, who values the prize not more than player $1$ is called the \emph{weaker player} and player $1$ is called the \emph{stronger player}.
Each player $i \in \{1,2\}$, not observing the choice of other player, chooses an effort level $x_i \in \mathbb{Z}_{\geq 0}$ in competition for the prize.
The player choosing the higher effort level wins the prize and, in the cases of a tie in effort levels, one of the players receives the prize with
probability $1/2$. Both player pay the price equal to their chosen effort levels.
Payoff to player $i$ from effort the pair of effort choices $(x_1,x_2) \in \mathbb{Z}^2_{\geq 0}$ is equal to
\begin{equation}
p^i(x_1,x_2) = \begin{cases}
               v_i - x_i, & \textrm{if $x_i > x_{-i}$,}\\
               \frac{v_i}{2} - x_i, & \textrm{if $x_i = x_{-i}$,}\\
               - x_i, & \textrm{if $x_i < x_{-i}$,}\\
               \end{cases}
\end{equation}
where $x_{-i}$ denotes the effort level chosen by the other player.
We allow the players to make randomize choices, so that each player $i$ chooses a probability distribution on non-negative integer numbers 
$\xi_i \in \Delta(\mathbb{Z}_{\geq 0})$.\footnote{
Given a set $S$, $\Delta(S)$ denotes the set of all probability distributions on $S$.}
For simplicity and notational convenience, with a probability distribution on $\mathbb{Z}_{\geq 0}$, $\xi$, we will identify a non-negative integer valued
random variable $X_i$ distributed according to $\xi_i$, so that for each $k \in \mathbb{Z}_{\geq 0}$, $\Pr(X_i = k) = \xi^i_k$.
We will also use the random variables to refer to the associated probability distributions.
Expected payoff to player $i$ from randomized effort choices $(X_1,X_2) \in \Delta(\mathbb{Z}_{\geq 0})^2$ is equal to
\begin{equation}
P^i(X_1,X_2) = v_i \Pr(X_i > X_{-i}) + \frac{v_i}{2} \Pr(X_1 = X_2) - \Ex(X_i)
\end{equation}
We assume that the players are risk neutral and each of them aims to maximise her expected payoff.
We are interested in mixed strategy Nash equilibria of this game, called Nash equilibria or equilibria, for short, throughout the paper.

\section{The analysis}
\label{sec:analysis}
Payoff to player $i$ from strategy profile $(X_1,X_2)$ can be written as
\begin{equation}
\label{eq:pandh}
\begin{aligned}
P^i(X_1,X_2) & = v_i\Pr(X_i > X_{-i}) + \frac{v_i}{2}\Pr(X_i = X_{-i}) - \Ex(X_i) \\
       & = \frac{v_i}{2}\Pr(X_i \geq X_{-i}) + \frac{v_i}{2}\Pr(X_i > X_{-i}) - \Ex(X_i) \\
       & = \frac{v_i}{2}\left(\Pr(X_i > X_{-i}) + 1 - \Pr(X_i < X_{-i})\right) - \Ex(X_i) \\
       & = \frac{v_i}{2}\left(\Pr(X_i > X_{-i}) - \Pr(X_i < X_{-i})\right) + \frac{v_i}{2} - \Ex(X_i) \\
       & = \frac{v_i}{2}\left(H(X_i,X_{-i}) - \left(\frac{2\Ex(X_i)}{v_i} - 1\right)\right),
\end{aligned}
\end{equation}
where
\begin{equation*}
H(X_i,X_{-i}) = \Pr(X_i > X_{-i}) - \Pr(X_i < X_{-i}).
\end{equation*}

Given probability distributions $(X_1,X_2)$, the quantity $H(X_i,X_{-i})$ is payoff to the player choosing $X_i$ against the choice $X_{-i}$ of the other
player is the \emph{discrete General Lotto game} defined in~\cite{Hart08}.
The game is played by two players, $1$ and $2$, who simultaneously and independently chooses probability distributions on non-negative integers. 
Each player $i$ is characterized by a number $b_i$ where $b_1\geq b_2 > 0$. Choices of player $i$ are constrained so that the player
chooses probability distributions $X_i$ with $\Ex(X_i) = b_i$.

The connection between continuous all-pay auctions and continuous General Lotto games is well known in the literature and complete characterisation
of equilibria in continuous all-pay auctions, obtained by~\cite{BayeKovenockVries96} was used by~\cite{Myerson93} and~\cite{SahuguetPersico06} to obtain
characterisation of equilibria in continuous General Lotto games.
In the case of discrete all-pay auctions we proceed in the reverse direction and use the complete characterisation of equilibria obtained
by~\cite{Hart08} and~\cite{Dziubinski12} to obtain complete characterisation of equilibria in discrete all-pay auctions.

First, we establish that the set of equilibria in discrete all-pay auctions is a subset of equilibria in discrete General Lotto games with properly
chosen constraints $(b_1,b_2)$.

\begin{proposition}
\label{pr:allpaylotto}
If a strategy profile $(X,Y)$ is a Nash equilibrium of all pay auction then it is also a Nash equilibrium of the General Lotto game $\Gammait(\Ex(X),\Ex(Y))$.
\end{proposition}

\begin{proof}
Suppose that $(X,Y)$ is a Nash equilibrium of all pay auction and let $v_1$ be the valuation of player $1$ and $v_2$ be the valuation of player $2$.
Since $(X,Y)$ is a Nash equilibrium so, for any strategy $X'$ of player $1$ with $\Ex(X') = \Ex(X)$,
\begin{equation*}
P^1(X,Y) \geq P^1(X',Y).
\end{equation*}
By~\eqref{eq:pandh} and $\Ex(X') = \Ex(X)$ this is equivalent to
\begin{equation*}
H(X,Y) \geq H(X',Y).
\end{equation*}
Similarly, for any strategy $Y'$ of player $2$ with $\Ex(Y') = \Ex(Y)$,
\begin{equation*}
H(Y,X) \geq H(Y',X).
\end{equation*}
Since $H$ is the payoff function in the General Lotto game and strategies in $\Gammait(\Ex(X),\Ex(Y))$ of each player are restricted to distributions with the same expected value,
$(X,Y)$ is a Nash equilibrium of $\Gammait(\Ex(X),\Ex(Y))$.
\end{proof}

The set of equilibria in a discrete all-pay auction is (usually, i.e. for most values of $v_1$ and $v_2$) a proper subset of equilibria in the corresponding
General Lotto games. Before we provide the characterisation of this set, we introduce the probability distributions that are the building blocks
of equilibria in General Lotto games.

\cite{Hart08} defines the following probability distributions. Given $m \geq 1$, let
\begin{equation*}
\uo{m} := U(\{1,3,\ldots,2m-1\}) =  \sum_{i=1}^{m}\left( \frac{1}{m} \right)\uvec{2i-1},
\end{equation*}
and, given $m \geq 0$, let
\begin{equation*} 
\ue{m} := U(\{0,2,\ldots,2m\}) =  \sum_{i=1}^{m+1}\left( \frac{1}{m+1} \right)\uvec{2i},
\end{equation*}
where, given an integer $j$, $\uvec{j}$ is the Dirac's measure putting probability $1$ on $j$.
Distributions $\uo{m}$ and $\ue{m}$ are ``uniform on odd numbers'' and ``uniform on even numbers'', respectively.
We will use
\begin{equation*}
\mathcal{U}^m = \{\ue{m}, \uo{m}\}
\end{equation*}
to denote the set of these distributions. \cite{Dziubinski12} defines the following distributions. First, given $m \geq 1$, let
\begin{equation*}
\uoup{m}  := U(\{2,4,\ldots,2m-2\}) =  \sum_{i=1}^{m-1}\left( \frac{1}{m-1} \right)\uvec{2i},
\end{equation*}
which is a uniform distribution on even numbers from $2$ to $2m-2$. Given $m \geq 2$ and $1 \leq j \leq m-1$, let
\begin{equation*} 
\w{j}{m} := \left(\frac{1}{2m}\right)\uvec{0} + \sum_{i=1}^{j-1}\left( \frac{1}{m} \right)\uvec{2i} + \left(\frac{1}{2m}\right)\uvec{2j} + \sum_{i=j+1}^{m}\left( \frac{1}{m} \right)\uvec{2i-1}.
\end{equation*}
Each distribution $\w{j}{m}$ is the distribution $\uo{m}$ distorted at the first $2j+1$ positions with a $2$-moving average, so that 
$\Pr(\w{j}{m} = i) = (\Pr(\uo{j} = i-1) + \Pr(\uo{j} = i+1))/2$, for $0 \leq i \leq 2j$ (where $\Pr(\uo{j} = -1) = 0$).
We will use
\begin{displaymath}
\mathcal{W}^m = \{\w{1}{m},\ldots, \w{m-1}{m}\}
\end{displaymath}
to denote the set of distributions $\w{j}{m}$. Given $m \geq 1$ and $1 \leq j \leq m$, let 
\begin{equation*} 
\vd{j}{m} := \sum_{i=1}^{j-1}\left( \frac{2}{2m+1} \right)\uvec{2i-1} + \left(\frac{1}{2m+1}\right)\uvec{2j-1} + \sum_{i=j}^{m}\left( \frac{2}{2m+1} \right)\uvec{2i}.
\end{equation*}
Each distribution $\vd{j}{m}$ is the distribution $\ue{m}$ distorted at the first $2j$ positions with a $2$-moving average, so that 
$\Pr(\vd{j}{m} = i) = (\Pr(\ue{j-1} = i-1) + \Pr(\ue{j-1} = i+1))/2$, for $0 \leq i \leq 2j-1$ (where $\Pr(\ue{j-1} = -1) = 0$).
We will use
\begin{displaymath}
\mathcal{V}^m = \{\vd{1}{m},\ldots, \vd{m}{m}\}
\end{displaymath}
to denote the set of distributions $\vd{j}{m}$.

With these distributions in hand, we are ready to state our main results. We divide the characterisation of equilibria in discrete all-pay auctions
into two cases, covered by Theorems~\ref{th:int} and~\ref{th:nonint} below. The first is the case where half of the valuation of the prize by the second (weaker) 
player is an integer number and the second is the case where it is not an integer number and it is greater than $1$.\footnote{
For completeness, in the Appendix we provide an additional Theorem~\ref{th:small} which covers the case of $v_2/2 \in (0,1)$.
}

\begin{theorem}
\label{th:int}
Strategy profile $(X,Y)$ is a Nash equilibrium of all-pay auction with players valuations $v_1 \geq v_2$ and $v_2/2 \in \mathbb{Z}_{\geq 1}$ if and only if
\begin{enumerate}[(i)]
\item if $v_1 = v_2$ then 
\begin{equation*}
X = \alpha \uo{m+1} + (1-\alpha) \ue{m} \textrm{, }\quad Y = \beta \uo{m+1} + (1-\beta)\ue{m},
\end{equation*}
with $m = v_2/2 - 1$, $\alpha \in [0,1]$, and $\beta \in [0,1]$. Equilibrium payoffs are 
\begin{equation*}
P^1(X,Y) = 1-\beta\ \textrm{ and }\ P^2(Y,X) = 1-\alpha.
\end{equation*}
\label{th:int:1}

\item if $v_1 > v_2 = 2$ then
\begin{equation*}
X = \uo{1}, \quad Y = (1 - b) \uvec{0} + b\left(\lambda \uo{1} + (1-\lambda) \ue{1}\right), 
\end{equation*}
where $b\in (0,1]$ and
\begin{equation*}
\frac{4}{bv_1} - \frac{2}{b} + 1 \leq \lambda \leq \frac{4}{bv_1} - 1.
\end{equation*}
Equilibrium payoffs are
\begin{equation*}
P^1(X,Y) = v_1 -   \frac{bv_1}{2} - 1\ \textrm{ and }\ P^2(Y,X) = 0.
\end{equation*}
\label{th:int:2}

\item if $v_1 > v_2 \geq 3$ then 
\begin{equation*}
X = \uo{m}, \quad Y = \left(1 - \frac{b}{m} \right) \uvec{0} + \left(\frac{b}{m}\right) Z, 
\end{equation*}
where $m = v_2/2$, $b \in [v_2(v_2-2)/(2v_1),\min(m,v_2(v_2+2)/(2v_1))]$, and 
\begin{equation*}
Z = \lambdao\uo{m} + \lambdae\ue{m} + \lambdaoup \uoup{m} + \sum_{j=1}^{m-1} \lambda_j \w{j}{m} 
\end{equation*}
with
\begin{equation*}
\lambdao,\lambdae,\lambdaoup,\lambda_1,\ldots,\lambda_{m-1} \geq 0 \textrm{ and } \lambdao + \lambdae + \lambdaoup + \sum_{j = 1}^{m-1} \lambda_j = 1,
\end{equation*}
\begin{equation*}
\frac{\lambdaoup}{m-1} - \frac{\lambdae}{m+1} = \frac{v_2^2}{2v_1b} - 1,
\end{equation*}
and
\begin{equation*}
\lambdao \geq \left(\frac{v_2}{2b}\right)\left(\frac{v_2(v_2+2)}{2v_1} + b - v_2\right)
\end{equation*}
Equilibrium payoffs are
\begin{equation*}
P^1(X,Y) =  v_1 - \frac{b v_1}{v_2} - \frac{v_2}{2}\ \textrm{ and }\ P^2(Y,X) = 0.
\end{equation*}
\label{th:int:3}

\end{enumerate}

\end{theorem}

The first point of the theorem covers the symmetric case where both players value the prize equally. In this case each of the players uses a convex 
combination of the uniform probability distribution on even numbers from $0$ to $2m$ and the uniform probability distribution on odd numbers from $0$ to $2m + 1$,
where $m = v_2/2-1$. There are continuum of possible equilibrium payoffs and payoff of each player is equal to $1$ minus the probability with which
the opponent uses the uniform on odd numbers probability distribution. In particular, payoff of $0$ as well as payoff of $1$ is possible for each player,
depending on the strategy used by the opponent.

The second and the third points of the theorem cover the asymmetric case where the valuation of player $2$ is strictly smaller than the valuation of player $1$.
The second point covers the subcase where $m = v_2/2$ takes value $1$ and the third point covers the remaining subcases.
In each case the stronger player mixes uniformly on odd numbers between $1$ and $2m-1$ while the weaker player chooses effort $0$ with probability
$1-b/m$ and, with probability $b/m$, uses a strategy which picks positive effort levels with probability greater than $0$. 
Like in the case of the first point of the theorem, there is continuum of equilibria and continuum of equilibrium payoffs. 
The equilibrium expected payoff to the stronger player depends on the expected value, $b$, of the strategy chosen by the weaker player.
The expected equilibrium payoff of the weaker player is equal to $0$ in all the cases.

Although there is continuum of equilibria and continuum of equilibrium payoffs when valuation of the prize by the weaker player is an even number,
the equilibria exhibit the interchangeability property: if $(X,Y)$ and $(X',Y')$ are both equilibria of all-pay auction, $(X,Y')$ and $(X',Y)$
are equilibria as well. Thus so long as each player chooses an equilibrium strategy, none of them has an incentive to deviate to a different strategy.

The next theorem provides complete characterisation of equilibria in the case where the valuation of the prize by the weaker player is not an 
even (integer) number.

\begin{theorem}
\label{th:nonint}
Strategy profile $(X,Y)$ is a Nash equilibrium of all-pay auction with players valuations $v_1 \geq v_2 > 2$ and $v_2/2 \notin \mathbb{Z}$ if and only if
\begin{enumerate}[(i)]
\item if $\lfloor v_1/2 \rfloor = \lfloor v_2/2 \rfloor$ then $X = \lambda \uo{m} + (1-\lambda)\ue{m}$ and $Y = \kappa \uo{m} + (1-\kappa) \ue{m}$, with
      $m = \lfloor v_2/2 \rfloor$,
\begin{equation*}
\kappa = \frac{\lfloor v_1/2 \rfloor}{v_1/2}\left(\left\lceil \frac{v_1}{2} \right\rceil -\frac{v_1}{2}\right) \textrm{ and } \lambda = \frac{\lfloor v_2/2 \rfloor}{v_2/2}\left(\left\lceil \frac{v_2}{2} \right\rceil -\frac{v_2}{2}\right).
\end{equation*}

Equilibrium payoffs of the players are
\begin{equation*}
P^1(X,Y) = \frac{v_1}{2} - \left\lfloor\frac{v_2}{2}\right\rfloor\ \textrm{ and }\ P^2(Y,X) = \frac{v_2}{2} - \left\lfloor\frac{v_2}{2}\right\rfloor.
\end{equation*}
\label{th:nonint:1}

\item if $v_1/2 = \lfloor v_2/2 \rfloor + 1$ then $Y = \ue{m}$, with $m = \lfloor v_2/2 \rfloor$, and
\begin{align*}
X & = \lambdao ((1-\alpha) \uo{m} + \alpha \uo{m+1}) + \lambdae ((1-\alpha) \ue{m} + \alpha \uo{m+1}) + {}\\
  & \qquad \sum_{j = 1}^{m} \lambda_j \left(\alpha\delta \vd{j}{m} + \left(1 - \alpha\delta \right) \uo{m}\right) + {}\\
  & \qquad \sum_{j = 1}^{m} \kappa_j \left(\alpha\delta \vd{j}{m} + \left(1 - \alpha\delta \right) \ue{m}\right),
\end{align*}
with 
\begin{equation*}
\delta = \frac{2\left\lfloor\frac{v_2}{2}\right\rfloor + 1}{\left\lfloor\frac{v_2}{2}\right\rfloor + 1},
\end{equation*} 
$\alpha \in [0,1/\delta]$, $\lambdao, \lambdae, \lambda_1, \ldots, \lambda_m, \kappa_1, \ldots, \kappa_m \geq 0$, $\lambdao + \lambdae + \sum_{j = 1}^{m} \lambda_j + \sum_{j = 1}^{m} \kappa_j = 1$, and
\begin{equation*}
\lambdae + \sum_{i = 1}^m \kappa_i \frac{1 - \alpha\delta}{1-\alpha} = \frac{\left\lceil \frac{v_2}{2} \right\rceil\left(\frac{v_2}{2} - \left\lfloor \frac{v_2}{2} \right\rfloor\right)}{\frac{v_2}{2}(1-\alpha)} - \frac{\alpha}{1-\alpha}
\end{equation*}
or 
\begin{align*}
X & = \lambdao ((1-\alpha) \uo{m} + \alpha \uo{m+1}) + \lambdae ((1-\alpha) \ue{m} + \alpha \uo{m+1}) + {}\\
  & \qquad \sum_{j = 1}^{m} \lambda_j \left((1-\alpha)\sigma \vd{j}{m} + \left(1 - (1-\alpha)\sigma \right) \uo{m+1}\right),
\end{align*}
with 
\begin{equation*}
\sigma = \frac{2\left\lfloor\frac{v_2}{2}\right\rfloor + 1}{\left\lfloor\frac{v_2}{2}\right\rfloor},
\end{equation*}
$\alpha \in (1/\delta,\frac{\left\lceil \frac{v_2}{2} \right\rceil}{\frac{v_2}{2}} \left(\frac{v_2}{2} - \left\lfloor \frac{v_2}{2} \right\rfloor\right)]$, $\lambdao, \lambdae, \lambda_1, \ldots, \lambda_m \geq 0$, $\lambdao + \lambdae + \sum_{j = 1}^{m} \lambda_j = 1$, and
\begin{equation*}
\lambdae = \frac{\left\lceil \frac{v_2}{2} \right\rceil\left(\frac{v_2}{2} - \left\lfloor \frac{v_2}{2} \right\rfloor\right)}{\frac{v_2}{2}(1-\alpha)} - \frac{\alpha}{1-\alpha}
\end{equation*}

Equilibrium payoffs of the players are
\begin{equation*}
P^1(X,Y) = 1\ \textrm{ and }\ P^2(Y,X) = 1-\frac{v_2}{v_1} \alpha - \frac{v_1-v_2}{2}.
\end{equation*}
\label{th:nonint:2}

\item if $v_1/2 > \lfloor v_2/2 \rfloor + 1$ then 
\begin{equation*}
X \in \conv(\{U^{m,\alpha}\} \cup \mathcal{X}^{m,\alpha}), \quad Y = \left(1 - \frac{b}{m} \right) \uvec{0} + \left(\frac{b}{m}\right) \ue{m}, 
\end{equation*}
where
\begin{equation*}
m = \left\lfloor \frac{v_2}{2} \right\rfloor,\quad b = \frac{\left\lfloor \frac{v_2}{2} \right\rfloor \left\lceil \frac{v_2}{2} \right\rceil}{\frac{v_1}{2}}, \quad \alpha = \frac{\left\lceil \frac{v_2}{2} \right\rceil}{\frac{v_2}{2}} \left(\frac{v_2}{2} - \left\lfloor \frac{v_2}{2} \right\rfloor \right),
\end{equation*}
\begin{itemize}
\item $U^{m,\alpha} = (1-\alpha)\uo{m} + \alpha \uo{m+1}$,
\end{itemize}
and
\begin{itemize}
\item $\mathcal{X}^{m,\alpha} = \alpha\delta \mathcal{V}^{m} + \left(1 - \alpha\delta \right) \uo{m}$, if $v_2/2 \leq \lceil v_2/2 \rceil - 1/2$,
\item $\mathcal{X}^{m,\alpha} = (1-\alpha)\sigma \mathcal{V}^{m} + \left(1 - (1-\alpha)\sigma \right) \uo{m+1}$, if $v_2/2 > \lceil v_2/2 \rceil - 1/2$,
\end{itemize}
where
\begin{equation*}
\delta = \frac{2\left\lfloor\frac{v_2}{2}\right\rfloor + 1}{\left\lfloor\frac{v_2}{2}\right\rfloor + 1},\quad \sigma = \frac{2\left\lfloor\frac{v_2}{2}\right\rfloor + 1}{\left\lfloor\frac{v_2}{2}\right\rfloor}.
\end{equation*}

Equilibrium payoffs of the players are
\begin{equation*}
P^1(X,Y) = v_1+1 - 2\left\lceil\frac{v_2}{2}\right\rceil\ \textrm{ and }\ P^2(Y,X) = 0.
\end{equation*}
\label{th:nonint:3}

\end{enumerate}
\end{theorem}

The first point of the theorem covers the case where valuations of the prize for the two players are close to each other:
the difference between them is less than $1$ and the closest integer value not greater than half of each valuation is the same for both of them.
Similarly to the first point of Theorem~\ref{th:int}, each of the players uses a convex combination of the uniform probability distribution on even 
numbers from $0$ to $2m$ and the uniform probability distribution on odd numbers from $0$ to $2m + 1$, where $m = \lfloor v_2/2 \rfloor$.
This time, however, equilibrium is unique. 

The second and the third point of the theorem cover the case where floors of the valuations of the prize for the two players are not equal.
The second point covers the case where half of the valuation of the stronger player is the smallest integer number greater than half of the valuation 
of the weaker player. The weaker player has a unique equilibrium strategy: mixing uniformly on even numbers between $0$ and $2\lfloor v_2/2 \rfloor$.
The stronger player has continuum of equilibrium strategies depending on the fraction $\alpha$ by which the expected value of the 
strategy of the stronger player exceeds the expected value of the equilibrium strategy of the weaker player, $\lfloor v_2/2 \rfloor$.
Payoff of the stronger player is equal to $1$ for all equilibria. Payoff of the weaker player is positive unless $\alpha$ attains its highest
value.
The third point covers the case where half of the valuation of the stronger player exceed the ceiling of the half the valuation of the weaker player.
In this case there is a unique equilibrium strategy of the weaker player: the player chooses effort $0$ with probability $1-b/m$ and, with probability $b/m$, mixes uniformly
on even numbers from $0$ to $2\lfloor v_2/2 \rfloor$. The stronger player has a continuum of equilibrium strategies. All the equilibria are payoff equivalent
and so equilibrium payoffs are unique. Equilibrium payoff of the weaker player is $0$ and the stronger player obtains a positive equilibrium payoff.

Notice that if the space of possible prize valuations, $v_1$ and $v_2$, is a subset of real numbers then equilibrium payoffs are generically unique: the cases where
equilibrium payoffs are not unique require one of the players to have a prize valuation which is an even number.

\subsection{Comparison with continuous all-pay auction}
In this section we compare equilibrium characterisation in discrete case with equilibrium characterisation in the continuous case.
The following result, stated in~\cite{Hillman88} and~\cite{HillmanRiley89} and rigorously proven in~\cite{BayeKovenockVries96},
provides full characterisation of equilibria for continuous all-pay auction.

\begin{theorem}[Hillman and Riley]
A strategy profile $(X,Y)$ is a Nash equilibrium of all-pay auction with continuous strategies and players valuations $v_1 \geq v_2 > 0$
if and only if $X$ is distributed uniformly on the interval $[0,v_2]$ while $Y$ is distributed on $[0,v_2]$ with a distribution with CDF
$F_2(x) = (x-v_2)/v_1 + 1$. Equilibrium payoffs of the players are $P_{\mathrm{cont.}}^1(X,Y) = v_1-v_2$ and $P_{\mathrm{cont.}}^2(X,Y) = 0$.
\end{theorem}

One feature of equilibrium strategies that is present in both the discrete and the continuous case is that the weaker player exerts zero
effort with probability close to $(v_1 - v_2)/2$ and then mixes with the remaining probability on the interval close to $[0,v_2]$.
In the case of the discrete model the interval is $[0,v_2]$ when $v_2$ is even and $[0,2\lfloor v_2/2\rfloor+1]$ when $v_2$ is not an even number.
The probability distribution on the interval, in the discrete case, is not uniform on discrete values, in general. However it is a convex combination
of probability distributions which are uniform on odd number, uniform on even number, or distorted such uniform distributions.
In particular in the case of the valuation of the weaker player, $v_2$, that is not even and is less than the valuation of the stronger player by
more than $2$ (point~\eqref{th:nonint:3} of Theorem~\ref{th:nonint}) the set of equilibrium strategies of the stronger player contains
probability distribution that is uniform on integer values in the interval $[0,2\lfloor v_2/2\rfloor+1]$.

\subsubsection{Comparative statics}
Fixing the valuation of player $1$, we analyse the effect of increasing the valuation of player $2$. 
Consider first the case when the valuation of player $1$, $v_1$, is not an even number, illustrated in Figure~\ref{fig:compstat12}. In this case equilibrium payoffs of player $2$ are unique for all values of $v_2$.
Notice that the difference in payoffs under the discrete and continuous case is equal to $0$ when $v_2 \leq 2\lfloor v_1/2 \rfloor$,
equal to $v_2/2 - \lfloor v_2/2 \rfloor$ when $2\lfloor v_1/2 \rfloor < v_2 \leq v_1$, equal to $v_1 - \lfloor v_1/2 \rfloor - v_2/2$ when $v_1 < v_2 \leq 2\lceil v_1/2 \rceil$,
and equal to $2(v_1/2 - \lfloor v_1/2 \rfloor - 1/2)$ when $v_2 > 2\lceil v_1/2 \rceil$. In particular, discreteness of strategy space benefits player 
$2$ when she is weaker but her valuation is close to the valuation of player $1$: $v_2 \in (2\lfloor v_1/2 \rfloor, v_1]$. Depending on whether 
the valuation of player $1$ is smaller or greater than the closest odd integer number, the discreteness of strategy space disbenefits or benefits player $2$, respectively,
when she is stronger than player $1$, $v_2 > v_1$.

\begin{figure}[!htb]
\begin{center}
 \includegraphics[scale=0.4]{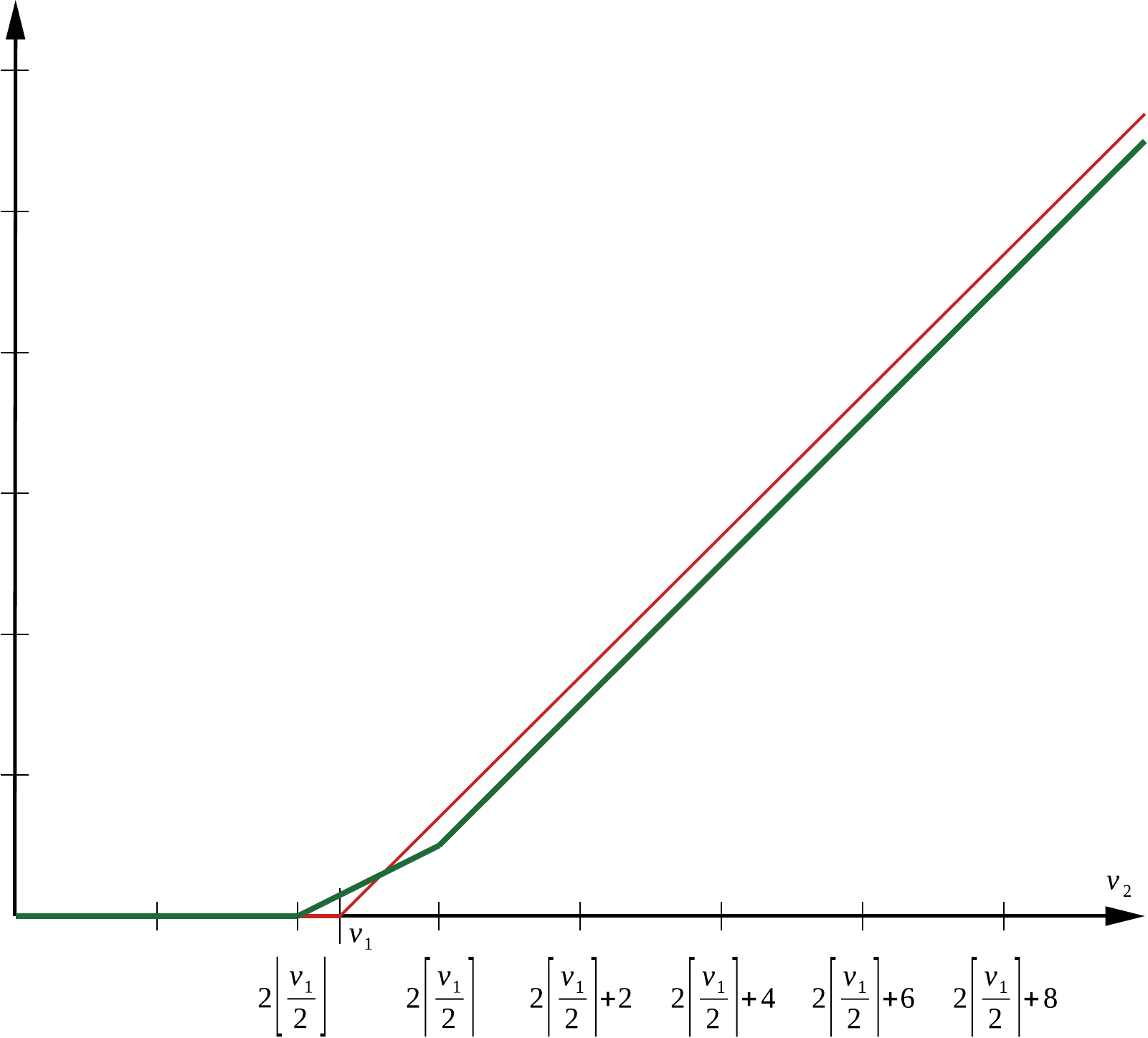}\hspace{5mm}
 \includegraphics[scale=0.4]{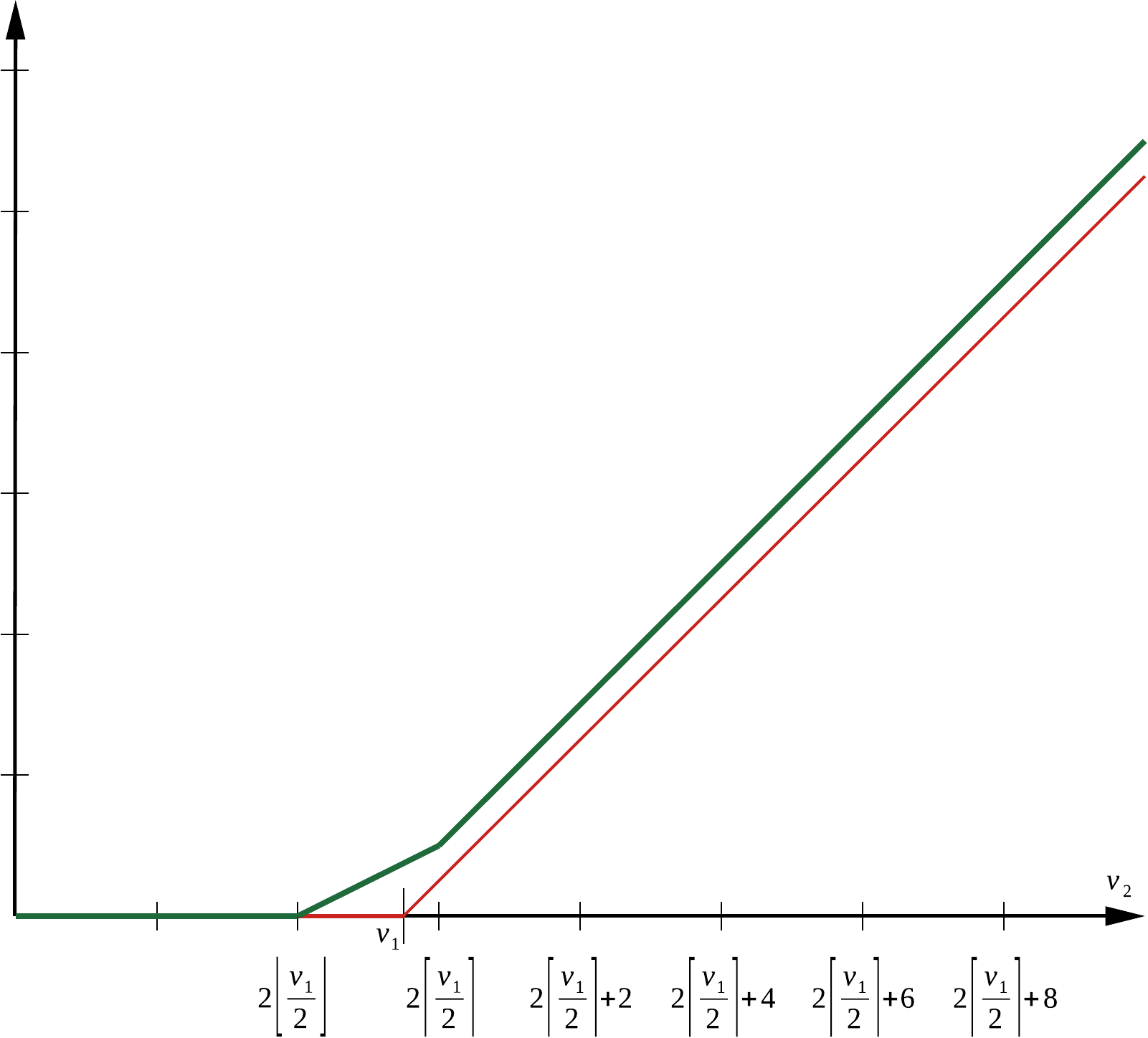}
 \caption{Change in payoff of player $2$ when $v_2$ increases: the case of $v_1$ not an even number. The case when $v_1/2 - \lfloor v_1/2 < 1/2$ (left) and the case of $v_1/2 - \lfloor v_1/2 \rfloor > 1/2$ (right).
 Thick line represents equilibrium payoffs in the discrete case and thin line represents equilibrium payoffs in continuous case.}
  \label{fig:compstat12}
\end{center}
\end{figure}

Second consider the case when valuation of player $1$, $v_1$, is an even number, illustrated in Figure~\ref{fig:compstat3}. In this case payoffs to player $2$ under the discrete and the continuous case
are equal when $v_2 \leq 2\lfloor v_1/2 \rfloor$. When $v_2 > 2\lfloor v_1/2 \rfloor$ there is a continuum of possible equilibrium payoffs and, depending on 
the strategy chosen by player $1$, player $2$ obtains lower or higher payoff under discrete strategy space as compared to the payoff under the continous case.

\begin{figure}[!htb]
\begin{center}
 \includegraphics[scale=0.4]{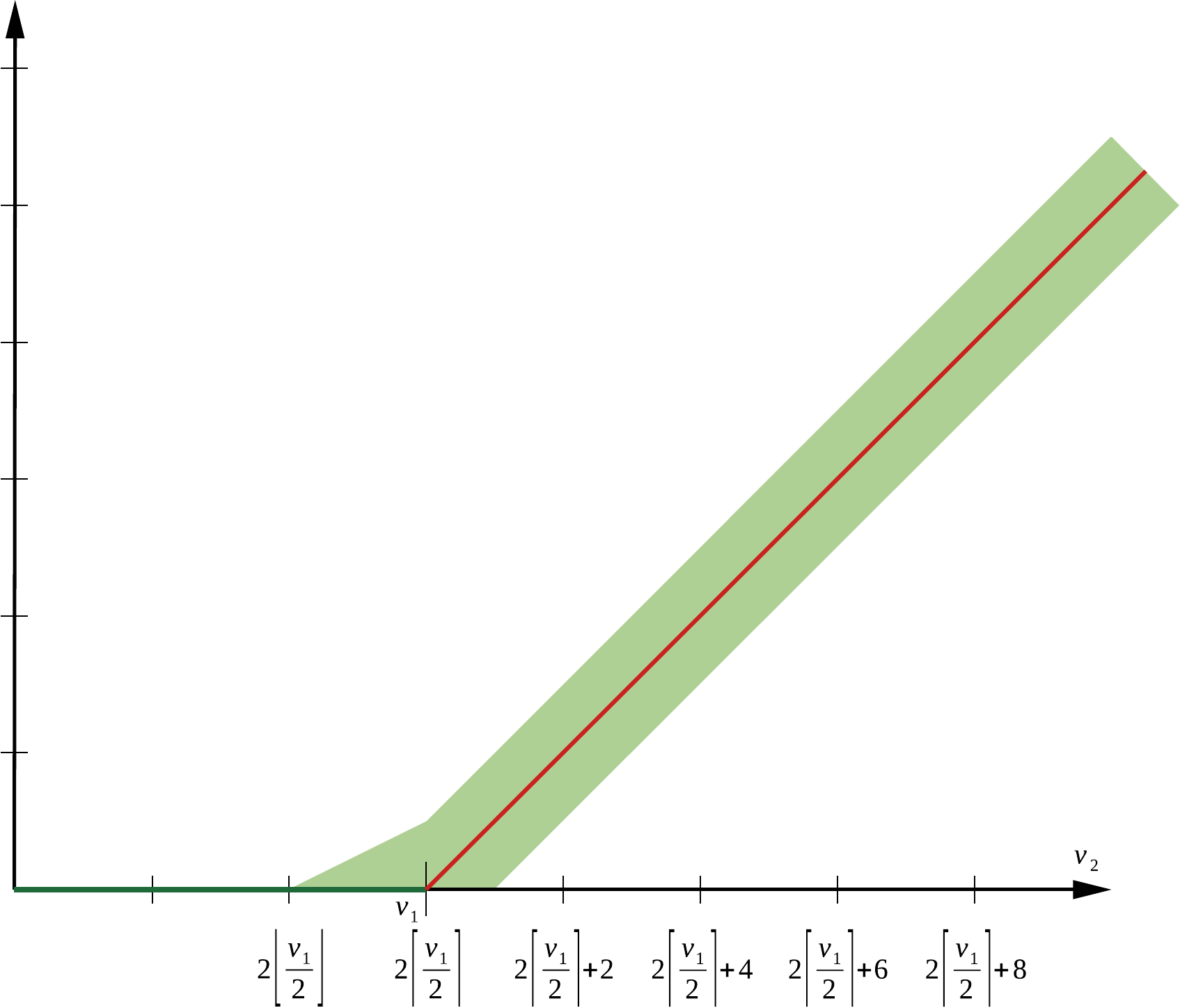}
 \caption{Change in payoff of player $2$ when $v_2$ increases: the case of $v_1$ an even number.
  Thin line represents payoffs in continuous case and the are around it represents the set of equilibrium payoffs in the discrete case.}
  \label{fig:compstat3}
\end{center}
\end{figure}

\section{Conclusions}
\label{sec:conclusions}
In this paper we studied two player all-pay auctions with discrete strategies of the players. We provided full characterisation of equilibria
as well as of equilibrium payoffs. We discussed how they are related to equilibria and equilibrium payoffs in the continuous variant of all-pay auctions.
We show that equilibria in the discrete variant are not unique, in general, but the are interchangeable. Equilibrium payoffs are unique, as long as none of
the players has valuation of the prize that is an even number. In case it is for at least one of the players, there is continuum of possible equilibrium payoffs.
Equilibrium strategies involve convex combinations of uniform distributions on even or odd numbers as well as distorted versions of such distributions.

\bibliographystyle{splncs04}
\bibliography{biblio}

\newpage
\appendix
\section{Appendix: Proofs}

Before giving the proofs of the main results, we provide some useful properties of function $H$.
In particular we provide bounds on its values when one of its arguments is a probability distribution
constituting a building block for equilibrium strategies.

For any strategies $X$ and $Y$,
\begin{equation}
\label{eq:hprop:1}
H(X,Y) = -H(Y,X)
\end{equation}
and for any strategies $X_1$, $X_2$, $Y$ and any $\lambda \in [0,1]$,
\begin{equation}
\label{eq:hprop:2}
H(\lambda X_1 + (1 - \lambda) X_2,Y) = \lambda H(X_1,Y) + (1-\lambda) H(X_2,Y).
\end{equation}
In addition, by~\eqref{eq:hprop:2},
\begin{multline}
\label{eq:pprop:1}
P(\lambda X_1 + (1-\lambda)X_2,Y) \\ \quad
\begin{aligned}
& {} =  \frac{v}{2}\left(H(\lambda X_1 + (1-\lambda) X_2,Y) - \left(\frac{2\Ex(\lambda X_1+ (1-\lambda) X_2)}{v} - 1\right)\right) \\
& {} =  \frac{v}{2}\left(\lambda H(X_1,Y) + (1-\lambda) H(X_2,Y) - \left(\frac{2(\lambda\Ex(X_1) + (1-\lambda) \Ex(X_2))}{v} - (\lambda + 1 - \lambda)\right)\right) \\
& {} = \lambda P(X_1,Y) + (1-\lambda)P(X_2,Y)
\end{aligned}
\end{multline}

As was shown in~\cite{Hart08},
\begin{equation}
\label{eq:h}
H(X,Y) = 1 - \sum_{i = 0}^{+\infty} \Pr(X = i)\left(\Pr(Y \geq i) + \Pr(Y \geq i+1) \right).
\end{equation}
By~\eqref{eq:h},
\begin{equation}
\label{eq:huvec0}
H(\uvec{0},Y) = \Pr(Y = 0) - 1.
\end{equation}
We will use the following inequalities, shown in~\cite{Hart08}: for any $Y$
\begin{equation}
\label{eq:huo}
H(\uo{m},Y) = 1 - \left(\frac{1}{m}\right)\sum_{i=1}^{2m} \Pr(Y \geq i) \geq 1 - \frac{\Ex(Y)}{m},
\end{equation}
with equality if and only if $\sum_{j = 2m+1}^{+\infty} \Pr(Y \geq j) = 0$, and
\begin{equation}
\label{eq:hue}
H(\ue{m},Y) = 1 - \left(\frac{1}{m+1}\right)\left(1 + \sum_{i=1}^{2m+1} \Pr(Y \geq i)\right) \geq 1 - \frac{\Ex(Y)+1}{m+1},
\end{equation}
with equality if and only if $\sum_{j = 2m+2}^{+\infty} \Pr(Y \geq j) = 0$. 
In addition, we will use the following inequalities, shown in~\cite{Dziubinski12}: for every $Y$, with $p_i = \Pr(Y = i)$:
\begin{equation}
\label{eq:hw}
H(W_j^m,Y) \geq 1 - \frac{\Ex(Y)}{m} + \frac{p_{2j} - p_0}{2m}
\end{equation}
with equality if and only if $\sum_{j = 2m+1}^{+\infty} \Pr(Y \geq j) = 0$,
\begin{equation}
\label{eq:huo1}
H(U_{\mathrm{O}\uparrow 1}^m,Y) \geq 1 - \frac{\Ex(Y) - 1}{m-1} - \frac{p_0}{m-1}
\end{equation}
with equality if and only if $\sum_{j = 2m}^{+\infty} \Pr(Y \geq j) = 0$, and
\begin{equation}
\label{eq:hv}
H(V_j^m,Y) \geq 1 - \frac{2\Ex(Y)}{2m+1} + \frac{p_{2j-1}}{2m+1}
\end{equation}
with equality if and only if $\sum_{j = 2m+2}^{+\infty} \Pr(Y \geq j) = 0$.

Using~\eqref{eq:huo} and~\eqref{eq:hprop:2}, for any $\alpha \in [0,1]$,
\begin{equation}
\label{eq:uouo}
H((1 - \alpha) \uo{m} + \alpha \uo{m+1},Y) \geq 1 - \Ex(Y)\left(\frac{1}{m+1} + \frac{1-\alpha}{m(m+1)}\right)
\end{equation}
with equality if and only if $\sum_{j = 2m+1}^{+\infty} \Pr(Y \geq j) = 0$.

Using~\eqref{eq:huo}, \eqref{eq:hue}, and~\eqref{eq:hprop:2}, for any $\alpha \in [0,1]$,
\begin{multline}
\label{eq:ueuo}
H((1 - \alpha) \ue{m} + \alpha \uo{m+1},Y) \geq 1 - \frac{\Ex(Y)+1}{m+1} + \frac{\alpha}{m+1} = {} \\
                                           1 - \Ex(Y)\left(\frac{1}{m+1} + \frac{1-\alpha}{m(m+1)}\right) + (1-\alpha)\frac{\Ex(Y) - m}{m(m+1)}
\end{multline}
with equality if and only if $\sum_{j = 2m+2}^{+\infty} \Pr(Y \geq j) = 0$.

Using~\eqref{eq:huo}, \eqref{eq:hv}, and~\eqref{eq:hprop:2}, for any $\alpha \in [0,1]$ and $\delta = (2m+1)/(m+1)$,
\begin{equation}
\label{eq:uov1}
H(\alpha \delta V_j^m + (1-\alpha \delta) \uo{m},Y) \geq 1 - \Ex(Y)\left(\frac{1}{m+1} + \frac{1-\alpha}{m(m+1)}\right) + \frac{\alpha}{m+1} p_{2j-1}
\end{equation}
with equality if and only if $\sum_{j = 2m+1}^{+\infty} \Pr(Y \geq j) = 0$.

Using~\eqref{eq:hue}, \eqref{eq:hv}, and~\eqref{eq:hprop:2}, for any $\alpha \in [0,1]$ and $\delta = (2m+1)/(m+1)$,
\begin{multline}
\label{eq:uev}
H(\alpha \delta V_j^m + (1-\alpha \delta) \ue{m},Y) \geq 1 - \frac{\Ex(Y)+1}{m+1}\left(1 + \frac{\alpha}{m+1}\right) + \frac{\alpha}{m+1} \left(2 + p_{2j-1} \right) = {}\\
                      1 - \Ex(Y)\left(\frac{1}{m+1} + \frac{1-\alpha}{m(m+1)}\right) + \frac{\alpha}{m+1}p_{2j-1} +
                      \frac{\Ex(Y)-m}{m+1}\left(\frac{1-\alpha}{m} - \frac{\alpha}{m+1}\right)
\end{multline}
with equality if and only if $\sum_{j = 2m+2}^{+\infty} \Pr(Y \geq j) = 0$.

Using~\eqref{eq:huo}, \eqref{eq:hv}, and~\eqref{eq:hprop:2}, for any $\alpha \in [0,1]$ and $\sigma = (2m+1)/m$,
\begin{equation}
\label{eq:uov2}
H((1-\alpha) \sigma V_j^m + (1-(1-\alpha) \sigma) \uo{m+1},Y) \geq 1 - \Ex(Y)\left(\frac{1}{m+1} + \frac{1-\alpha}{m(m+1)}\right) + \frac{1-\alpha}{m} p_{2j-1}
\end{equation}
with equality if and only if $\sum_{j = 2m+2}^{+\infty} \Pr(Y \geq j) = 0$.

The rest of the proof section is organized as follows.
First we obtain necessary and sufficient properties of equilibria of discrete all-pay auctions for the following cases of expected values 
of the equilibrium strategies $X$ and $Y$:
\begin{itemize}
\item $\Ex(X) = \Ex(Y) = m$ with $m \in \mathbb{Z}_{\geq 1}$ (Proposition~\ref{pr:5}),
\item $\Ex(X) = m+\alpha$ and $\Ex(Y) = m+\beta$ with $m \in \mathbb{Z}_{\geq 0}$ and $\alpha,\beta \in (0,1)$ (Proposition~\ref{pr:6}),
\item $\Ex(X) = m$ and $\Ex(Y) = b$ with $m \in \mathbb{Z}_{\geq 1}$ and $m > b > 0$ (Propositions~\ref{pr:7:1} and~\ref{pr:7:2}),
\item $\Ex(X) = m+\alpha$ and $\Ex(Y) = b$ with $m\in \mathbb{Z}_{\geq 1}$, $m > b > 0$, and $\alpha \in (0,1)$ (Propositions~\ref{pr:8:1} and~\ref{pr:8:2}).
\end{itemize}

These cases cover all possible configurations of expected values of the strategies in equilibrium strategy profiles.
In the last section of the appendix, we use these characterisations to provide proofs of the main theorems~\ref{th:int} and~\ref{th:nonint}.

\subsection*{The case of $\Ex(X) = \Ex(Y) = m$ with $m \in \mathbb{Z}_{\geq 1}$}

\begin{proposition}
\label{pr:5}
A strategy profile $(X,Y)$ such that $\Ex(X) = \Ex(Y) = m$, $m \in \mathbb{Z}_{\geq 1}$, is a Nash equilibrium of all pay auction with both players valuations $v_1\geq v_2 > 0$ if and only if, 
\begin{itemize}
\item Either $m = \lfloor v_2/2 \rfloor = \lfloor v_1/2 \rfloor$ or $m = \lceil v_1/2 \rceil - 1 = \lceil v_1/2 \rceil - 1$ or $m = \lfloor v_2/2 \rfloor = \lceil v_1/2 \rceil - 1$,
\item 
\begin{equation*}
X = \lambda \uo{m} + (1-\lambda)\ue{m}, \qquad Y = \kappa \uo{m} + (1-\kappa) \ue{m}.
\end{equation*}
\item
\begin{equation*}
\kappa = \frac{2m}{v_1}\left(m+1-\frac{v_1}{2}\right) \textrm{ and } \lambda = \frac{2m}{v_2}\left(m+1-\frac{v_2}{2}\right).
\end{equation*}

Equilibrium payoffs of the players are
\begin{equation*}
P^1(X,Y) = \frac{v_1}{2} - \left\lfloor\frac{v_2}{2}\right\rfloor\ \textrm{ and }\ P^2(Y,X) = \frac{v_2}{2} - \left\lfloor\frac{v_2}{2}\right\rfloor.
\end{equation*}

\end{itemize}
\end{proposition}

\begin{proof}
For the left to right implication, suppose that $(X,Y)$ satisfying the condition stated in the theorem is a Nash equilibrium of all pay auction with both players valuations $v_1,v_2 > 0$. Then, by Proposition~\ref{pr:allpaylotto}, it is a Nash equilibrium of general Lotto game $\Gamma(m,m)$.
Hence, by~\cite[Theorem~2]{Hart08},
\begin{equation}
\label{eq:lemma:5:1}
X = \lambda \uo{m} + (1-\lambda)\ue{m}, \qquad Y = \kappa \uo{m} + (1-\kappa) \ue{m},
\end{equation}
with $\lambda,\kappa \in [0,1]$, and
\begin{equation}
\label{eq:lemma:5:2}
H(X,Y) = 0.
\end{equation}

Given $c \in (-m,m)$ let
\begin{equation*}
X' = \gamma \uvec{\lfloor m+c \rfloor} + (1 - \gamma)\uvec{\lceil m+c \rceil},
\end{equation*}
where
\begin{equation*}
\gamma = \begin{cases}
     1, & \textrm{if $m+c \in \mathbb{Z}$}\\
     \frac{\lceil m+c \rceil - (m+c)}{\lceil m+c \rceil - \lfloor m+c \rfloor}, & \textrm{otherwise}
     \end{cases}
\end{equation*}
so that $\Ex(X') = m + c$.
Since $c \in (-m,m)$ and $m \in \mathbb{Z}_{\geq 1}$, so $\lfloor m+c \rfloor \geq 0$ and $\lceil m+c \rceil \leq 2m$.
Hence $\Pr(X' \geq 2m+1) = 0$. Thus, by~\eqref{eq:lemma:5:1}, \eqref{eq:hprop:1}, and~\eqref{eq:hprop:2},
\begin{equation}
\label{eq:lemma:5:3}
\begin{aligned}
H(X',Y) & = -\kappa H(\uo{m},X') - (1-\kappa) H(\ue{m},X')\\
        & = \kappa \left(\frac{\Ex(X')}{m} - 1\right) + (1-\kappa) \left(\frac{\Ex(X')+1}{m+1} - 1\right)\\
        & = \Ex(X')\frac{m+\kappa}{m(m+1)} + \frac{1-\kappa}{m+1} - 1\\
        & = (m+c)\frac{m+\kappa}{m(m+1)} + \frac{1-\kappa}{m+1} - 1.
\end{aligned}        
\end{equation}

By~\eqref{eq:pandh}, \eqref{eq:lemma:5:2}, and~\eqref{eq:lemma:5:3}
\begin{align*}
P^1(X',Y) & = \frac{v_1}{2}\left((m+c)\frac{m+\kappa}{m(m+1)} + \frac{1-\kappa}{m+1} - \frac{2(m+c)}{v_1}\right), \textrm{ and}\\
P^1(X,Y) & = \frac{v_1}{2}\left(1 - \frac{2m}{v_1}\right).
\end{align*}
Since $(X,Y)$ is a Nash equilibrium so $P(X,Y) \geq P(X',Y)$. Hence
\begin{equation*}
\frac{v_1}{2}\left(1 - \frac{2m}{v_1}\right) \geq \frac{v_1}{2}\left((m+c)\frac{m+\kappa}{m(m+1)} + \frac{1-\kappa}{m+1} - \frac{2(m+c)}{v_1}\right)
\end{equation*}
and, since $v_1 > 0$, so
\begin{equation}
\label{eq:lemma:5:4}
\frac{2c}{v_1} \geq \frac{c(m+\kappa)}{m(m+1)}.
\end{equation}
Since $v_1 > 0$, $m \geq 0$, $\kappa \in (0,1)$, and~\eqref{eq:lemma:5:4} holds for any $c \in (-m,m)$ so
$v_1/2 = m(m+1)/(m+\kappa)$. By analogous derivation we also get $v_2/2 = m(m+1)/(m+\lambda)$.
From this we also get
\begin{equation*}
\kappa = \frac{2m}{v_1}\left(m+1-\frac{v_1}{2}\right) \textrm{ and } \lambda = \frac{2m}{v_2}\left(m+1-\frac{v_2}{2}\right).
\end{equation*}

Since $\kappa \in [0,1]$ so $m(m+1)/(m+\kappa) \in [m,m+1]$. Hence either $m = \lfloor v_1/2\rfloor$ or $m = \lceil v_1/2 \rceil - 1$.
Similarly, either $m = \lfloor v_2/2\rfloor$ or $m = \lceil v_2/2 \rceil - 1$. 

For the left to right implication, consider a strategy $X'$ with $\Ex(X')\geq 0$ of player $1$. By~\eqref{eq:lemma:5:3} and~\eqref{eq:pandh},
\begin{align*}
P^1(X',Y) & = \frac{v_1}{2}\left(\Ex(X')\frac{m+\kappa}{m(m+1)} + \frac{1-\kappa}{m+1} - 1 - \left(\frac{2\Ex(X')}{v_1} - 1\right)\right)\\
        & = \frac{v_1}{2}\left(\frac{1-\kappa}{m+1} + \Ex(X')\left(\frac{m+\kappa}{m(m+1)} - \frac{2}{v_1}\right)\right).
\end{align*}
Similarly, by~\eqref{eq:lemma:5:2} and~\eqref{eq:pandh},
\begin{equation*}
P^1(X,Y) = \frac{v_1}{2}\left(1 - \frac{2m}{v_1}\right) = \frac{v_1}{2}\left(\frac{1-\kappa}{m+1} + m\left(\frac{m+\kappa}{m(m+1)} - \frac{2}{v_1}\right)\right).
\end{equation*}
Since $v_1/2 = m(m+1)/(m+\kappa)$ so $P^1(X',Y) = 0 = P^1(X,Y)$ and there is no profitable deviation for player $1$ from $(X,Y)$.
By analogous derivation, using $v_2/2 = m(m+1)/(m+\lambda)$, we conclude that there is no profitable deviation for player $2$ from $(X,Y)$ either. Hence $(X,Y)$ is a Nash equilibrium.
\end{proof}

\subsection*{The case of $\Ex(X) = m+\alpha$ and $\Ex(Y) = m+\beta$ with $m \in \mathbb{Z}_{\geq 0}$ and $\alpha,\beta \in (0,1)$}

\begin{proposition}
\label{pr:6}
A strategy profile $(X,Y)$ such that $\Ex(X) = m + \alpha$, $\Ex(Y) = m+\beta$, $m \in \mathbb{Z}_{\geq 0}$, $0 < \beta \leq \alpha < 1$, is a Nash equilibrium of all pay auction with both players valuations $v_1,v_2 > 0$ if and only if, $v_1/2 = v_2/2 = m+1$ and
\begin{equation*}
X = (1 - \alpha) \ue{m} + \alpha\uo{m+1}, \qquad Y = (1 - \beta) \ue{m} + \beta\uo{m+1}.
\end{equation*}
Equilibrium payoffs of the players are
\begin{equation*}
P^1(X,Y) = 1-\beta\ \textrm{ and }\ P^2(Y,X) = 1-\alpha.
\end{equation*}

\end{proposition}

\begin{proof}
For the left to right implication, suppose that $(X,Y)$ satisfying the condition stated in the theorem is a Nash equilibrium of all pay auction with both players valuations $v_1,v_2 > 0$. Then, by Proposition~\ref{pr:allpaylotto}, it is a Nash equilibrium of general Lotto game $\Gamma(m+\alpha,m+\beta)$.
Hence, by~\cite[Theorem~3]{Hart08},
\begin{equation}
\label{eq:lemma:6:1}
X = (1 - \alpha) \ue{m} + \alpha\uo{m+1}, \qquad Y = (1 - \beta) \ue{m} + \beta\uo{m+1},
\end{equation}
and
\begin{equation}
\label{eq:lemma:6:2}
H(X,Y) = \frac{\alpha - \beta}{m+1}.
\end{equation}

Given $c \in (-m-\alpha,m+1-\alpha)$ let
\begin{equation*}
X' = \gamma \uvec{\lfloor m+\alpha+c \rfloor} + (1 - \gamma)\uvec{\lceil m+\alpha+c \rceil},
\end{equation*}
where
\begin{equation*}
\gamma = \begin{cases}
     1, & \textrm{if $m+c+\alpha \in \mathbb{Z}$}\\
     \frac{\lceil m+c+\alpha \rceil - (m+c+\alpha)}{\lceil m+c+\alpha \rceil - \lfloor m+c+\alpha \rfloor}, & \textrm{otherwise}
     \end{cases}
\end{equation*}
so that $\Ex(X') = m + \alpha + c$.
Since $c \in (-m-\alpha,m+1-\alpha)$, $m \in \mathbb{Z}_{\geq 0}$, and $\alpha \in (0,1)$ so $\lfloor m+c+\alpha \rfloor \geq 0$ and $\lceil m+c+\alpha \rceil \leq 2m+1$.
Hence $\Pr(X' \geq 2m+2) = 0$. Thus, by~\eqref{eq:lemma:6:1}, \eqref{eq:hprop:1}, and~\eqref{eq:hprop:2},
\begin{equation}
\label{eq:lemma:6:3}
\begin{aligned}
H(X',Y) & = -\left(1 - \beta \right) H(\ue{m},X') - \beta H(\uo{m+1},X')\\
        & = \left(1 - \beta \right) \left(\frac{\Ex(X')+1}{m+1} - 1\right) + \beta \left(\frac{\Ex(X')}{m+1} - 1\right)\\
        & = \frac{m+\alpha+c}{m+1} + \frac{1-\beta}{m+1} - 1 = \frac{\alpha-\beta+c}{m+1}.
\end{aligned}        
\end{equation}

By~\eqref{eq:pandh}, \eqref{eq:lemma:6:2}, and~\eqref{eq:lemma:6:3}
\begin{align*}
P^1(X',Y) & = \frac{v_1}{2}\left(\frac{\alpha-\beta+c}{m+1} - \left(\frac{2(m+c+\alpha)}{v_1} - 1\right)\right), \textrm{ and}\\
P^1(X,Y) & = \frac{v_1}{2}\left(\frac{\alpha - \beta}{m+1} - \left(\frac{2(m+\alpha)}{v_1} - 1\right)\right).
\end{align*}
Since $(X,Y)$ is a Nash equilibrium so $P(X,Y) \geq P(X',Y)$. Hence
\begin{equation*}
\frac{v_1}{2}\left(\frac{\alpha - \beta}{m+1} - \left(\frac{2(m+\alpha)}{v_1} - 1\right)\right) \geq
 \frac{v_1}{2}\left(\frac{\alpha-\beta+c}{m+1} - \left(\frac{2(m+c+\alpha)}{v_1} - 1\right)\right)
\end{equation*}
and, since $v_1 > 0$, so
\begin{equation}
\label{eq:lemma:6:4}
\frac{2c}{v_1} \geq \frac{c}{m+1}.
\end{equation}
Since $v_1 > 0$, $m \geq 0$, $\alpha \in (0,1)$, and~\eqref{eq:lemma:6:4} holds for any $c \in (-m-\alpha,m+1-\alpha)$ so
$v_1/2 = m+1$. By analogous derivation we also get $v_2/2 = m+1$.

For the left to right implication, consider a strategy $X'$ with $\Ex(X')\geq 0$ of player $1$. By~\eqref{eq:lemma:6:3} and~\eqref{eq:pandh},
\begin{align*}
P^1(X',Y) & = \frac{v_1}{2}\left(\frac{\Ex(X')}{m+1} + \frac{1-\beta}{m+1} - 1 - \left(\frac{2\Ex(X')}{v_1} - 1\right)\right)\\
        & = \frac{v_1}{2}\left(\frac{1-\beta}{m+1} + \Ex(X')\left(\frac{1}{m+1} - \frac{2}{v_1}\right)\right).
\end{align*}
Similarly, by~\eqref{eq:lemma:6:2} and~\eqref{eq:pandh},
\begin{align*}
P^1(X,Y) & = \frac{v_1}{2}\left(\frac{\alpha - \beta}{m+1} - \left(\frac{2(m+\alpha)}{v_1} - 1\right)\right)\\
       & = \frac{v_1}{2}\left(\frac{1 - \beta}{m+1} + (m+\alpha)\left(\frac{1}{m+1} - \frac{2}{v_1}\right)\right)
\end{align*}
Since $v_1/2 = m+1$ so $P^1(X',Y) = 0 = P^1(X,Y)$ and there is no profitable deviation for player $1$ from $(X,Y)$.
By analogous derivation, using $v_2/2 = m+1$, we conclude that there is no profitable deviation for player $2$ from $(X,Y)$ either. Hence $(X,Y)$ is a Nash equilibrium.
\end{proof}

\subsection*{The case of $\Ex(X) = m$ and $\Ex(Y) = b$ with $m \in \mathbb{Z}_{\geq 1}$ and $m > b > 0$}

\begin{proposition}
\label{pr:7:1}
A strategy profile $(X,Y)$ such that $\Ex(X) = m = 1$, $\Ex(Y) = b \in (0,1)$, is a Nash equilibrium of all pay auction with both players valuations $v_1,v_2 > 0$ if and only if,
\begin{equation*}
v_1\geq v_2 = 2,\quad b \in \left(0, \frac{4}{v_1} \right]
\end{equation*}
\begin{equation*}
X = \uo{1}, \quad Y = (1 - b) \uvec{0} + b\left(\lambda \uo{1} + (1-\lambda) \ue{1}\right), 
\end{equation*}
where
\begin{equation*}
\frac{4}{bv_1} - \frac{2}{b} + 1 \leq \lambda \leq \frac{4}{bv_1} - 1.
\end{equation*}

Equilibrium payoffs of the players are
\begin{equation*}
P^1(X,Y) = \frac{v_1}{2}(2 - b) - 1\ \textrm{ and }\ P^2(Y,X) = 0.
\end{equation*}

\end{proposition}

\begin{proposition}
\label{pr:7:2}
A strategy profile $(X,Y)$ such that $\Ex(X) = m$, $\Ex(Y) = b$,  $m\in \mathbb{Z}_{\geq 2}$, and $m > b > 0$, is a Nash equilibrium of all pay auction with both players valuations $v_1,v_2 > 0$ if and only if,
\begin{equation*}
v_1\geq v_2, \quad m = \frac{v_2}{2},\quad b \in \left[\frac{v_2(v_2-2)}{2v_1},\frac{v_2(v_2 + 2)}{2v_1}\right],
\end{equation*}
\begin{equation*}
X = \uo{m}, \quad Y = \left(1 - \frac{b}{m} \right) \uvec{0} + \left(\frac{b}{m}\right) Z, 
\end{equation*}
where
\begin{equation*}
Z = \lambdao\uo{m} + \lambdae\ue{m} + \lambdaoup \uoup{m} + \sum_{j=1}^{m-1} \lambda_j \w{j}{m} 
\end{equation*}
with
\begin{equation*}
\lambdao,\lambdae,\lambdaoup,\lambda_1,\ldots,\lambda_{m-1} \geq 0 \textrm{ and } \lambdao + \lambdae + \lambdaoup + \sum_{j = 1}^{m-1} \lambda_j = 1,
\end{equation*}
\begin{equation*}
\frac{\lambdaoup}{m-1} - \frac{\lambdae}{m+1} = \frac{v_2^2}{2v_1b} - 1,
\end{equation*}
and
\begin{equation*}
\lambdao \geq \left(\frac{v_2}{2b}\right)\left(\frac{v_2(v_2+2)}{2v_1} + b - v_2\right)
\end{equation*}

Equilibrium payoffs of the players are
\begin{equation*}
P^1(X,Y) = \frac{v_1-v_2}{2} + \frac{v_1}{v_2}\left(\frac{v_2}{2} - b\right)\ \textrm{ and }\ P^2(Y,X) = 0.
\end{equation*}

\end{proposition}

\begin{lemma}
\label{lemma:7:1}
If a strategy profile $(X,Y)$ such that $\Ex(X) = m$, $\Ex(Y) = b$,  $m\in \mathbb{Z}_{\geq 1}$, and $m > b > 0$, is a Nash equilibrium of all pay auction with both players valuations $v_1,v_2 > 0$ then $m = v_2/2$, $(X,Y)$ is a Nash equilibrium of $\Gamma(m,b)$, and $X = \uo{m}$. 
\end{lemma}

\begin{proof}
Since $(X,Y)$ is a Nash equilibrium of all pay auction so, by Proposition~\ref{pr:allpaylotto}, it is a Nash equilibrium of general Lotto game $\Gamma(m,b)$.
Hence, by~\cite[Theorem~2]{Hart08},
\begin{equation}
\label{eq:lemma:7:1:1}
X = \uo{m}
\end{equation}
and
\begin{equation}
\label{eq:lemma:7:1:2}
H(X,Y) = 1-\frac{b}{m}.
\end{equation}

Given $c \in (-b,2m-b)$ let
\begin{equation*}
Y' = \gamma \uvec{\lfloor b+c \rfloor} + (1 - \gamma)\uvec{\lceil b+c \rceil},
\end{equation*}
where
\begin{equation*}
\gamma = \begin{cases}
     1, & \textrm{if $b+c \in \mathbb{Z}$}\\
     \frac{\lceil b+c \rceil - (b+c)}{\lceil b+c \rceil - \lfloor b+c \rfloor}, & \textrm{otherwise}
     \end{cases}
\end{equation*}
so that $\Ex(X') = b + c$.
Since $c \in (-b,2m-b)$ and $m \in \mathbb{Z}_{\geq 1}$ so $\lceil b+c \rceil \leq 2m$ and $\Pr(X' \geq 2m+1) = 0$. Thus, by~\eqref{eq:lemma:7:1:1}, \eqref{eq:huo}, and~\eqref{eq:hprop:1},
\begin{equation}
\label{eq:lemma:7:1:3}
H(Y',X) = \frac{b+c}{m}-1.
\end{equation}
By~\eqref{eq:lemma:7:1:2}, \eqref{eq:lemma:7:1:3}, and~\eqref{eq:pandh},
\begin{align*}
P(Y',X) & = \frac{v_2}{2}(b+c)\left(\frac{1}{m} - \frac{2}{v_2}\right), \textrm{ and}\\
P(Y,X) & = \frac{v_2}{2}b\left(\frac{1}{m} - \frac{2}{v_2}\right)
\end{align*}

Since $(X,Y)$ is a Nash equilibrium so $P(Y,X) \geq P(Y',X)$. Hence
\begin{equation*}
\frac{v_2}{2}b\left(\frac{1}{m} - \frac{2}{v_2}\right) \geq \frac{v_2}{2}(b+c)\left(\frac{1}{m} - \frac{2}{v_2}\right)
\end{equation*}
and, consequently,
\begin{equation}
\label{eq:lemma:7:1:4}
\frac{2c}{v_2} \geq \frac{c}{m}.
\end{equation}
Since $v_2 > 0$, $b > 0$, $m > b$, and~\eqref{eq:lemma:7:1:4} holds for all $c \in (-b,2m-b)$ so $m = v_2/2$.
\end{proof}

\begin{lemma}
\label{lemma:7:2}
If a strategy profile $(X,Y)$ such that $\Ex(X) = m$, $\Ex(Y) = b$,  $m\in \mathbb{Z}_{\geq 1}$, and $m > b > 0$, is a Nash equilibrium of all pay auction with both players valuations $v_1,v_2 > 0$ then $Y = \left(1 - \frac{b}{m} \right) \uvec{0} + \left(\frac{b}{m}\right) Z$, where
\begin{itemize}
\item if $m = 1$ then
\begin{equation*}
Z = \lambdao\uo{m} + \lambdae\ue{m}
\end{equation*}
with
\begin{equation*}
\lambdao,\lambdae \geq 0 \textrm{, } \lambdao + \lambdae = 1,
\end{equation*}
\begin{equation}
\label{lemma:7:2:1}
\frac{\lambdae}{2} \geq 1 - \frac{2}{bv_1}.
\end{equation}
and
\begin{equation}
\label{lemma:7:2:2}
\frac{\lambdae}{2} \leq \frac{1}{b} - \frac{2}{bv_1}.
\end{equation}
\item if $m \geq 2$ then
\begin{equation*}
Z = \lambdao\uo{m} + \lambdae\ue{m} + \lambdaoup \uoup{m} + \sum_{j=1}^{m-1} \lambda_j \w{j}{m} 
\end{equation*}
with
\begin{equation}
\label{lemma:7:2:3}
\lambdao,\lambdae,\lambdaoup,\lambda_1,\ldots,\lambda_{m-1} \geq 0 \textrm{, } \lambdao + \lambdae + \lambdaoup + \sum_{j = 1}^{m-1} \lambda_j = 1,
\end{equation}
\begin{equation}
\label{lemma:7:2:4}
\frac{\lambdae}{m+1} + \frac{1}{2m}\sum_{j = 1}^{m-1} \lambda_j \leq \frac{m}{b}\left(1 - \frac{v_2}{v_1}\right),
\end{equation}
\begin{equation}
\label{lemma:7:2:5}
\frac{\lambdao}{m} + \frac{\lambdae}{m+1} + \frac{\lambdaoup}{m-1} + \frac{\sum_{j = 1}^{m-1} \lambda_j}{m} = \frac{2m}{b v_1}.
\end{equation}
and, in the case of $b > m-1$,
\begin{equation}
\label{lemma:7:2:6}
\frac{\sum_{j = 1}^{m-1} \lambda_j}{2m} + \frac{\lambdaoup}{m-1} \leq \frac{m-b}{b}.
\end{equation}
\end{itemize}

\end{lemma}

\begin{proof}
Take any strategy profile $(X,Y)$ such that $\Ex(X) = m$, $\Ex(Y) = b$,  $m\in \mathbb{Z}_{\geq 1}$, and $m \geq b > 0$, and suppose that it is a Nash equilibrium of all pay auction with both players valuations $v_1,v_2 > 0$.

Since $(X,Y)$ is a Nash equilibrium of all pay auction so, by Proposition~\ref{pr:allpaylotto}, it is a Nash equilibrium of general Lotto game $\Gamma(m,b)$.
Thus, by~\cite[Theorem~2]{Hart08} and~\cite[Theorems~3 and 4]{Dziubinski12},
\begin{equation}
\label{eq:lemma:7:2:3}
H(X,Y) = 1-\frac{b}{m},
\end{equation}
\begin{equation}
\label{eq:lemma:7:2:1}
X = \uo{m},
\end{equation}
\begin{equation}
\label{eq:lemma:7:2:2}
Y = \left(1 - \frac{b}{m} \right) \uvec{0} + \left(\frac{b}{m}\right) Z,
\end{equation}
with $Z \in \conv \left(\mathcal{U}^m \cup \mathcal{W}^m \cup \{\uoup{m}\}\right)$, in the case of $m \geq 2$, and $Z \in \conv\ \mathcal{U}^m$, in the case of $m = 1$.
In addition, in the case of $m \geq 2$ and $b > m-1$, by~\cite[Theorems~4]{Dziubinski12} (in particular (26) in proof of the theorem),
\begin{equation*}
\frac{\sum_{j = 1}^{m-1} \lambda_j}{2m} + \frac{\lambdaoup}{m-1} \leq \frac{m-b}{b}.
\end{equation*}

By~\eqref{eq:lemma:7:2:2},
\begin{equation*}
Y = \left(1 - \frac{b}{m} \right) \uvec{0} + \left(\frac{b}{m}\right) \left(\lambdao\uo{m} + \lambdae\ue{m}+\lambdaoup\uoup{m}   + \sum_{j = 1}^{m-1} \lambda_j \w{j}{m} \right),
\end{equation*}
where $\lambdao, \lambdae, \lambdaoup, \lambda_1, \ldots, \lambda_m \geq 0$, $\lambdao + \lambdae + \lambdaoup + \sum_{j = 1}^{m-1} \lambda_j = 1$,
and $\lambdaoup = 0$, in the case of $m = 1$.
In the case of $m \geq 2$, by~\eqref{eq:hprop:2}, \eqref{eq:huo}, \eqref{eq:hue}, \eqref{eq:hw}, and~\eqref{eq:huo1}, for any strategy of the first player, $X'$
\begin{equation}
\label{eq:lemma:7:2:4}
\begin{aligned}
H(Y,X') & = \left(1 - \frac{b}{m} \right) H(\uvec{0},X') + \frac{b}{m}\Bigg( \lambdao H(\uo{m}, X') + \lambdae H(\ue{m}, X') + \lambdaoup H(\uoup{m}, X') + {}\\
& \qquad \sum_{j = 1}^{m-1} \lambda_j H(\w{j}{m},X') \Bigg)\\
        & \geq \left(1 - \frac{b}{m} \right)(\Pr(X' = 0)-1) + \frac{b}{m}\Bigg(\lambdao \left(1 - \frac{\Ex(X')}{m}\right) \\
        & \quad {} + \lambdae \left( 1 - \frac{\Ex(X')+1}{m+1}\right) + \lambdaoup \left( 1 - \frac{\Ex(X')-1}{m-1} - \frac{\Pr(X'=0)}{m-1} \right) \\
        & \quad {} + \sum_{j = 1}^{m-1} \lambda_j \left(1 - \frac{\Ex(X')}{m} + \frac{\Pr(X' = 2j) - \Pr(X' = 0)}{2m}\right)\Bigg)\\
        & = -1 + \frac{b}{m} \Bigg(2 - \Ex(X')\left(\frac{\lambdao}{m} + \frac{\lambdae}{m+1} + \frac{\lambdaoup}{m-1} + \frac{\sum_{j = 1}^{m-1} \lambda_j}{m}\right) - \frac{\lambdae}{m+1} + \frac{\lambdaoup}{m-1} \Bigg) \\
        & \qquad {} + \left(1 - \frac{b}{m}\left(1 + \frac{\lambdaoup}{m-1} + \frac{\sum_{j = 1}^{m-1} \lambda_j}{2m}\right)\right) \Pr(X' = 0) + \frac{\sum_{j = 1}^{m-1} \lambda_j \Pr(X' = 2j)}{2m}
\end{aligned}
\end{equation}
with equality only if $\sum_{j = 2m}^{+\infty} \Pr(X' \geq j) = 0$.
In the case of $m = 1$, by~\eqref{eq:hprop:2}, \eqref{eq:huo}, and \eqref{eq:hue}, for any strategy of the first player, $X'$
\begin{equation}
\label{eq:lemma:7:2:5}
\begin{aligned}
H(Y,X') & = \left(1 - \frac{b}{m} \right) H(\uvec{0},X') + \frac{b}{m}\left( \lambdao H(\uo{m}, X') + \lambdae H(\ue{m}, X')\right)\\
        & \geq \left(1 - \frac{b}{m} \right)(\Pr(X' = 0)-1) + \frac{b}{m}\left(\lambdao \left(1 - \frac{\Ex(X')}{m}\right) + \lambdae \left( 1 - \frac{\Ex(X')+1}{m+1}\right)\right)\\
        & = -1 + \frac{b}{m} \left(2 - \Ex(X')\left(\frac{\lambdao}{m} + \frac{\lambdae}{m+1}\right) - \frac{\lambdae}{m+1}\right) + \left(1 - \frac{b}{m}\right) \Pr(X' = 0)
\end{aligned}
\end{equation}
with equality only if $\sum_{j = 2m+1}^{+\infty} \Pr(X' \geq j) = 0$.

Using~\eqref{eq:lemma:7:2:1}, \eqref{eq:lemma:7:2:4}, \eqref{eq:lemma:7:2:5}, \eqref{eq:pandh}, and~\eqref{eq:hprop:1},
in the case of $m \geq 2$,
\begin{equation}
\label{eq:lemma:7:2:6}
\begin{aligned}
P(X',Y) & = \frac{v_1}{2}\left(H(X',Y) - \left(\frac{2\Ex(X')}{v_1} - 1\right)\right)
          = \frac{v_1}{2}\left(1 - H(Y,X') - \frac{2\Ex(X')}{v_1}\right) \\
        & \leq \frac{v_1}{2}\Bigg(2 - \frac{b}{m} \Bigg(2 - \Ex(X')\left(\frac{\lambdao}{m} + \frac{\lambdae}{m+1} + \frac{\lambdaoup}{m-1} + \frac{\sum_{j = 1}^{m-1} \lambda_j}{m}\right) - \frac{\lambdae}{m+1} + \frac{\lambdaoup}{m-1} \Bigg) \\
        & \qquad {} - \left(1 - \frac{b}{m}\left(1 + \frac{\lambdaoup}{m-1} + \frac{\sum_{j = 1}^{m-1} \lambda_j}{2m}\right)\right) \Pr(X' = 0) \\
        & \qquad {} - \frac{\sum_{j = 1}^{m-1} \lambda_j \Pr(X' = 2j)}{2m} - \frac{2\Ex(X')}{v_1}\Bigg),
\end{aligned}
\end{equation}
with equality only if $\sum_{j = 2m}^{+\infty} \Pr(X' \geq j) = 0$, and in the case of $m = 1$,
\begin{equation}
\label{eq:lemma:7:2:7}
\begin{aligned}
P(X',Y) & = \frac{v_1}{2}\left(H(X',Y) - \left(\frac{2\Ex(X')}{v_1} - 1\right)\right)
          = \frac{v_1}{2}\left(1 - H(Y,X') - \frac{2\Ex(X')}{v_1}\right) \\
        & \leq \frac{v_1}{2}\Bigg(2 - \frac{b}{m} \Bigg(2 - \Ex(X')\left(\frac{\lambdao}{m} + \frac{\lambdae}{m+1}\right) - \frac{\lambdae}{m+1} \Bigg) \\
        & \qquad {} - \left(1 - \frac{b}{m}\right) \Pr(X' = 0) - \frac{2\Ex(X')}{v_1}\Bigg),
\end{aligned}
\end{equation}
with equality only if $\sum_{j = 2m+1}^{+\infty} \Pr(X' \geq j) = 0$.

On the other hand, by \eqref{eq:pandh} and~\eqref{eq:lemma:7:2:3},
\begin{align*}
P(X,Y) & = \frac{v_1}{2}\left(2 - \frac{b}{m} - \frac{2m}{v_1}\right).
\end{align*}

In the case of $m = 1$, let $X' = \uvec{2}$. Since $\Pr(X' = 0)= 0$ and $\sum_{j = 2m}^{+\infty} \Pr(Y' \geq j) = 0$ so, by~\eqref{eq:lemma:7:2:5},
\begin{equation*}
\begin{aligned}
P(X',Y) & = \frac{v_1}{2}\left(2 - \frac{b}{m} \left(2 - 2\left(\frac{\lambdao}{m} + \frac{\lambdae}{m+1}\right) - \frac{\lambdae}{m+1}\right) - \frac{4}{v_1}\right),
\end{aligned}
\end{equation*}
Since $m = 1$ and $\lambdao + \lambdae = 1$ so this can be rewritten as
\begin{equation*}
\begin{aligned}
P(X',Y) & = \frac{v_1}{2}\left(2 - b \left(2 - 2\left(1 - \frac{\lambdae}{2}\right) - \frac{\lambdae}{2}\right) - \frac{4}{v_1}\right) \\
        & = \frac{v_1}{2}\left(2 - b - \frac{2}{v_1} + b\left(1 - \frac{\lambdae}{2}\right) - \frac{2}{v_1}\right)
\end{aligned}
\end{equation*}
Since $(X,Y)$ is a Nash equilibrium so $P(X,Y) \geq P(X',Y)$ and so
\begin{equation*}
b\left(1 - \frac{\lambdae}{2}\right) - \frac{2}{v_1} \leq 0
\end{equation*}
from which it follows that
\begin{equation*}
\frac{\lambdae}{2} \geq 1 - \frac{2}{bv_1}.
\end{equation*}

In the case of $m \geq 2$, given $c \in [-(m-1), m-1]$ let
\begin{equation*}
X' = \gamma \uvec{2\lfloor \frac{m+c+1}{2}\rfloor -1} + (1-\gamma) \uvec{2\lceil \frac{m+c+1}{2}\rceil -1}
\end{equation*}
where
\begin{equation*}
\gamma = \begin{cases}
		 1, & \textrm{if $m+c+1$ is an even number},\\
         \frac{\lceil \frac{m+c+1}{2}\rceil - \frac{m+c+1}{2}}{\lceil \frac{m+c+1}{2}\rceil - \lfloor \frac{m+c+1}{2}\rfloor}, & \textrm{otherwise}.
         \end{cases}
\end{equation*}
Notice that $\Ex(X') = m+c$, $\Pr(X' = 2j) = 0$ for all $j \geq 0$, and, since $c \in [-(m-1), m-1]$, $\sum_{j = 2m}^{+\infty} \Pr(Y' \geq j) = 0$.
Hence, by~\eqref{eq:lemma:7:2:5},
\begin{equation*}
\begin{aligned}
P(X',Y) & = \frac{v_1}{2}\Bigg(2 - \frac{b}{m} \Bigg(2 - (m+c)\left(\frac{\lambdao}{m} + \frac{\lambdae}{m+1} + \frac{\lambdaoup}{m-1} + \frac{\sum_{j = 1}^{m-1} \lambda_j}{m}\right)\\
        & \qquad {} - \frac{\lambdae}{m+1} + \frac{\lambdaoup}{m-1} \Bigg) - \frac{2(m+c)}{v_1}\Bigg),
\end{aligned}
\end{equation*}
Since $(X,Y)$ is a Nash equilibrium so $P(X,Y) \geq P(X',Y)$ and so
\begin{multline*}
2 - \frac{b}{m} - \frac{2m}{v_1} \geq 2 - \frac{b}{m} \Bigg(2 - (m+c)\left(\frac{\lambdao}{m} + \frac{\lambdae}{m+1} + \frac{\lambdaoup}{m-1} + \frac{\sum_{j = 1}^{m-1} \lambda_j}{m}\right)\\
{}  - \frac{\lambdae}{m+1} + \frac{\lambdaoup}{m-1} \Bigg) - \frac{2(m+c)}{v_1}
\end{multline*}
which can be rewritten as
\begin{equation*}
\frac{2c}{v_1} + \frac{b}{m} \geq \frac{b(m+c)}{m}\left(\frac{\lambdao}{m} + \frac{\lambdae}{m+1} + \frac{\lambdaoup}{m-1} + \frac{\sum_{j = 1}^{m-1} \lambda_j}{m}\right) + \frac{b}{m} \left(\frac{\lambdae}{m+1} - \frac{\lambdaoup}{m-1}\right)
\end{equation*}
and further as
\begin{equation*}
\frac{2c}{v_1} + \frac{b}{m} \geq \frac{bc}{m}\left(\frac{\lambdao}{m} + \frac{\lambdae}{m+1} + \frac{\lambdaoup}{m-1} + \frac{\sum_{j = 1}^{m-1} \lambda_j}{m}\right) 
+ \frac{b}{m}\left(\lambdao + \lambdae + \lambdaoup + \sum_{j = 1}^{m-1} \lambda_j\right).
\end{equation*}
Since $\lambdao + \lambdae + \lambdaoup + \sum_{j = 1}^{m-1} \lambda_j = 1$ so this can be further rewritten as
\begin{equation*}
\frac{2c}{v_1} \geq \frac{bc}{m}\left(\frac{\lambdao}{m} + \frac{\lambdae}{m+1} + \frac{\lambdaoup}{m-1} + \frac{\sum_{j = 1}^{m-1} \lambda_j}{m}\right).
\end{equation*}
Since $v_1 > 0$, $b > 0$, $m > b$, and~\eqref{eq:lemma:7:1:4} holds for all $c \in (-b,2m-b)$ so
\begin{equation*}
\frac{\lambdao}{m} + \frac{\lambdae}{m+1} + \frac{\lambdaoup}{m-1} + \frac{\sum_{j = 1}^{m-1} \lambda_j}{m} = \frac{2m}{b v_1}.
\end{equation*}

Consider a strategy $X' = \uvec{0}$ of the first player. Since $\Ex(X' = 0)$, $\Pr(X' = 0)$, $\Pr(X' = 2j) = 0$ for all $j > 0$, and $\sum_{j = 2m}^{+\infty} \Pr(Y' \geq j) = 0$ so,
by~\eqref{eq:lemma:7:2:6}, in the case of $m \geq 2$
\begin{align*}
P(\uvec{0},Y) & = \frac{v_1}{2}\Bigg(2 - \frac{b}{m} \Bigg(2 - \frac{\lambdae}{m+1} + \frac{\lambdaoup}{m-1} \Bigg)
                  - 1 + \frac{b}{m}\left(1 + \frac{\lambdaoup}{m-1} + \frac{\sum_{j = 1}^{m-1} \lambda_j}{2m}\right)\Bigg)\\
			 & = \frac{v_1}{2}\left(1 - \frac{b}{m} + \frac{b}{m}\left(\frac{\lambdae}{m+1} + \frac{\sum_{j = 1}^{m-1} \lambda_j}{2m}\right)\right)
\end{align*}
and, by~\eqref{eq:lemma:7:2:7}, in the case of $m = 1$,
\begin{align*}
P(\uvec{0},Y) & = \frac{v_1}{2}\left(2 - \frac{b}{m} \left(2 - \frac{\lambdae}{m+1}\right) - 1 + \frac{b}{m}\right)\\
			 & = \frac{v_1}{2}\left(1 - \frac{b}{m} + \frac{b}{m}\left(\frac{\lambdae}{m+1}\right)\right)
\end{align*}

Since $(X,Y)$ is a Nash equilibrium so $P(X,Y) \geq P(X',Y)$ and so, in the case of $m \geq 2$,
\begin{equation*}
2 - \frac{b}{m} - \frac{2m}{v_1} \geq 1 - \frac{b}{m} + \frac{b}{m}\left(\frac{\lambdae}{m+1} + \frac{\sum_{j = 1}^{m-1} \lambda_j}{2m}\right)
\end{equation*}
from which (using $m = v_2/2$ from Lemma~\ref{lemma:7:1}) it follows that
\begin{equation*}
\frac{\lambdae}{m+1} + \frac{\sum_{j = 1}^{m-1} \lambda_j}{2m} \leq \frac{m}{b}\left(1 - \frac{v_2}{v_1}\right).
\end{equation*}
Similarly, in the case of $m = 1$ we get
\begin{equation*}
\frac{\lambdae}{m+1} \leq \frac{m}{b}\left(1 - \frac{v_2}{v_1}\right).
\end{equation*}
Since $m=1$ and, by Lemma~\ref{lemma:7:1}, $v_2 = 2m = 2$ so this can be rewritten as
\begin{equation*}
\frac{\lambdae}{2} \leq \frac{1}{b} - \frac{2}{bv_1}.
\end{equation*}
\end{proof}

\begin{lemma}
\label{lemma:7:3}
If conditions in Lemma~\ref{lemma:7:1} and~Lemma~\ref{lemma:7:2} are satisfied for a strategy profile $(X,Y)$ such that $\Ex(X) = m$, $\Ex(Y) = b$,  $m\in \mathbb{Z}_{\geq 1}$, and $m \geq b> 0$, then $(X,Y)$ is a Nash equilibrium of all pay auction with both players valuations $v_1,v_2 > 0$.
\end{lemma}

\begin{proof}
Take any strategy profile $(X,Y)$, as described in the lemma. By Lemma~\ref{lemma:7:1}, $(X,Y)$ is a Nash equilibrium of General Lotto game $\Gamma(m,b)$.
Hence, by~\cite[Theorem~2]{Hart08},
\begin{equation}
\label{eq:lemma:7:3:1}
H(X,Y) = 1-\frac{b}{m}.
\end{equation}

We show first that there is no profitable deviation for player $1$.
By~\eqref{eq:pandh} and~\eqref{eq:lemma:7:3:1},
\begin{equation}
\label{eq:lemma:7:3:2}
P(X,Y) = \frac{v_1}{2}\left(2 - \frac{b}{m} - \frac{2m}{v_1}\right).
\end{equation}
Take any strategy $X'$ of player $1$ with $\Ex(X') \geq 0$. We will consider the cases of $m \geq 2$ and $m = 1$ separately.

Suppose first that $m \geq 2$. By~\eqref{eq:lemma:7:2:6},
\begin{align*}
P(X',Y) & \leq \frac{v_1}{2}\Bigg(2 - \frac{b}{m} \Bigg(2 - \Ex(X')\left(\frac{\lambdao}{m} + \frac{\lambdae}{m+1} + \frac{\lambdaoup}{m-1} + \frac{\sum_{j = 1}^{m-1} \lambda_j}{m}\right) - \frac{\lambdae}{m+1} + \frac{\lambdaoup}{m-1} \Bigg) \\
        & \qquad {} - \left(1 - \frac{b}{m}\left(1 + \frac{\lambdaoup}{m-1} + \frac{\sum_{j = 1}^{m-1} \lambda_j}{2m}\right)\right) \Pr(X' = 0) - \frac{2\Ex(X')}{v_1}\Bigg).
\end{align*}
By Lemma~\ref{lemma:7:2}, \eqref{lemma:7:2:3} and~\eqref{lemma:7:2:5},
\begin{equation*}
\frac{1 - \lambdae - \lambdaoup}{m} + \frac{\lambdae}{m+1} + \frac{\lambdaoup}{m-1} = \frac{2m}{b v_1}
\end{equation*}
from which it follows that
\begin{equation}
\label{eq:lemma:7:3:3}
\frac{\lambdaoup}{m-1} - \frac{\lambdae}{m+1} = \frac{2m^2}{b v_1} - 1.
\end{equation}
By~\eqref{eq:lemma:7:1:1} and Lemma~\ref{lemma:7:2}, \eqref{lemma:7:2:5},
\begin{align*}
P(X',Y) & \leq \frac{v_1}{2}\Bigg(2 - \frac{b}{m} \left(2 - \Ex(X')\left(\frac{2m}{b v_1}\right) + \frac{2m^2}{b v_1} - 1 \right) \\
        & \qquad {} - \left(1 - \frac{b}{m}\left(1 + \frac{\lambdaoup}{m-1} + \frac{\sum_{j = 1}^{m-1} \lambda_j}{2m}\right)\right) \Pr(X' = 0) - \frac{2\Ex(X')}{v_1}\Bigg) \\
        & = \frac{v_1}{2}\Bigg(2 - \frac{b}{m} - \frac{2m}{v_1} 
        - \left(1 - \frac{b}{m}\left(1 + \frac{\lambdaoup}{m-1} + \frac{\sum_{j = 1}^{m-1} \lambda_j}{2m}\right)\right) \Pr(X' = 0)\Bigg) \\
\end{align*}
Lemma~\ref{lemma:7:2}, \eqref{lemma:7:2:4}, together with~\eqref{eq:lemma:7:3:3} yields
\begin{equation*}
\frac{b}{m}\left(1 + \frac{\lambdaoup}{m-1} + \frac{1}{2m}\sum_{j = 1}^{m-1} \lambda_j\right) \leq 1 - \frac{v_2}{v_1} + \frac{2m}{v_1}
\end{equation*}
which, together with $m = v_2/2$ (from Lemma~\ref{lemma:7:1}), yields
\begin{equation*}
1 - \frac{b}{m}\left(1 + \frac{\lambdaoup}{m-1} + \frac{1}{2m}\sum_{j = 1}^{m-1} \lambda_j\right) \geq 0.
\end{equation*}
Hence
\begin{align*}
P(X',Y) & \leq \frac{v_1}{2}\left(2 - \frac{b}{m} - \frac{2m}{v_1}\right) = P(X,Y)
\end{align*}
and so player $1$ has no profitable deviation from $(X,Y)$.

Next, suppose that $m = 1$. Since $\Ex(X') = \sum_{j \geq 1} \Pr(X' \geq j) \geq \Pr(X' \geq 1) = 1 - \Pr(X' = 0)$ so
$\Pr(X' = 0) \geq \max(0,1 - \Ex(X'))$. This, together with~\eqref{eq:lemma:7:2:7} and $m = 1$ yields,
\begin{align*}
P(X',Y) & \leq \frac{v_1}{2}\Bigg(2 - b \Bigg(2 - \Ex(X')\left(\lambdao + \frac{\lambdae}{2}\right) - \frac{\lambdae}{2} \Bigg) \\
        & \qquad {} - \left(1 - b\right) \max(0,1-\Ex(X')) - \frac{2\Ex(X')}{v_1}\Bigg),
\end{align*}
By Lemma~\ref{lemma:7:2}, $\lambdao + \lambdae = 1$. Hence
\begin{equation*}
P(X',Y) \leq \frac{v_1}{2}\left(2 - b - \frac{2}{v_1} + (\Ex(X')-1)\left(b\left(1 - \frac{\lambdae}{2}\right) - \frac{2}{v_1}\right) - \left(1 - b\right) \max(0,1-\Ex(X'))\right),
\end{equation*}
Suppose that $\Ex(X') \geq 1$. Then, by Lemma~\ref{lemma:7:2}, \eqref{lemma:7:2:1},
\begin{equation*}
b\left(1 - \frac{\lambdae}{2}\right) - \frac{2}{v_1} \leq 0
\end{equation*}
and so
\begin{align*}
P(X',Y) \leq \frac{v_1}{2}\left(2 - b - \frac{2}{v_1}\right) = P(X,Y).
\end{align*}
Suppose that $\Ex(X') < 1$. Then, 
\begin{equation*}
P(X',Y) \leq \frac{v_1}{2}\left(2 - b - \frac{2}{v_1} - (1 - \Ex(X'))\left(b\left(1 - \frac{\lambdae}{2}\right) - \frac{2}{v_1} + 1 - b\right)\right)
\end{equation*}
and, by Lemma~\ref{lemma:7:2}, \eqref{lemma:7:2:2}
\begin{equation*}
b\left(1-\frac{\lambdae}{2}\right) - \frac{2}{v_1} + 1-b \geq 0.
\end{equation*}
Hence
\begin{equation*}
P(X',Y) \leq \frac{v_1}{2}\left(2 - b - \frac{2}{v_1}\right) = P(X,Y).
\end{equation*}
Thus player $1$ has no profitable deviation from $(X,Y)$.

Second, we show that there is no profitable deviation for player $2$. By Lemma~\ref{lemma:7:1}, $m = v_2/2$. Thus, by~\eqref{eq:lemma:7:3:1}, \eqref{eq:hprop:1}, \eqref{eq:pandh},
\begin{align*}
P(Y,X) = \frac{v_2}{2}b\left(\frac{1}{m} - \frac{2}{v_2}\right) = 0.
\end{align*}
Take any strategy $Y'$ of player $2$ with $\Ex(Y') \geq 0$. By Lemma~\ref{lemma:7:1}, $X = \uo{m}$ and, by~\eqref{eq:huo} and~\eqref{eq:hprop:1},
\begin{equation*}
H(Y',X) \leq \frac{\Ex(Y')}{m} - 1.
\end{equation*}
Thus, by~\eqref{eq:pandh} and $m = v_2/2$,
\begin{align*}
P(Y',X) & = \frac{v_2}{2}\left(H(Y',X) - \left(\frac{2\Ex(Y')}{v_2} - 1\right)\right)\\
        & \leq \frac{v_2}{2}\left(\frac{\Ex(Y')}{m} - \frac{2\Ex(Y')}{v_2}\right) = 0 = P(Y,X).
\end{align*}
Hence player $2$ has no profitable deviation from $(X,Y)$.
\end{proof}

\begin{proof}[Proof of Proposition~\ref{pr:7:1}]
Lemmas~\ref{lemma:7:1}, \ref{lemma:7:2}, and \ref{lemma:7:3} establish the sufficient and necessary properties of equilibrium strategies.
The condition on $v_1$, $v_2$, and $b$, follow from the lemmas. By Lemma~\ref{lemma:7:1}, $v_2 = 2m = 2$. By Lemma~\ref{lemma:7:2},
$\lambda \in [0,1]$ and
\begin{equation*}
\frac{1 - \lambda}{2} \leq \frac{1}{b} - \frac{2}{bv_1}.
\end{equation*}
Hence
\begin{equation*}
\frac{1}{b} - \frac{2}{bv_1} = \frac{v_2 - 2}{b v_1} \geq 0
\end{equation*}
and so $v_1 \geq 2 = v_2$.
Similarly, by Lemma~\ref{lemma:7:2},
\begin{equation*}
\frac{1 - \lambda}{2} \geq 1 - \frac{2}{bv_1}.
\end{equation*}
Hence
\begin{equation*}
\frac{bv_1-2}{bv_1} \leq \frac{1}{2}
\end{equation*}
from which it follows that $b \leq 4/v_1$.
\end{proof}

\begin{proof}[Proof of Proposition~\ref{pr:7:2}]
Lemmas~\ref{lemma:7:1}, \ref{lemma:7:2}, and \ref{lemma:7:3} establish the sufficient and necessary properties of equilibrium strategies.

By Lemma~\ref{lemma:7:2}, 
\begin{equation}
\label{eq:pr:7:2:4}
\lambdao + \lambdae + \lambdaoup + \sum_{j = 1}^{m-1} \lambda_j = 1
\end{equation}
and
\begin{equation}
\label{eq:pr:7:2:3}
\frac{\lambdao}{m} + \frac{\lambdae}{m+1} + \frac{\lambdaoup}{m-1} + \frac{\sum_{j = 1}^{m-1} \lambda_j}{m} = \frac{2m}{b v_1}.
\end{equation}
This, together with $m = v_2/2$, by Lemma~\ref{lemma:7:1}, yields
\begin{equation}
\label{eq:pr:7:2:2}
\frac{\lambdaoup}{m-1} - \frac{\lambdae}{m+1} = \frac{v_2^2}{2 b v_1} - 1.
\end{equation}

Note that~\eqref{eq:pr:7:2:2} together with
\begin{equation*}
\frac{\lambdae}{m+1} + \frac{1}{2m}\sum_{j = 1}^{m-1} \lambda_j \leq \frac{v_2}{2b}\left(1 - \frac{v_2}{v_1}\right),
\end{equation*}
holds if and only if~\eqref{eq:pr:7:2:2} together with 
\begin{equation}
\label{eq:pr:7:2:5}
\frac{\lambdaoup}{m-1} + \frac{1}{2m}\sum_{j = 1}^{m-1} \lambda_j \leq \frac{v_2}{2b}-1
\end{equation}
holds.

Using~\eqref{eq:pr:7:2:4}, \eqref{eq:pr:7:2:2}, and $m = v_2/2$, the left hand side of~\eqref{eq:pr:7:2:5} can be rewritten as
\begin{align*}
\frac{\lambdaoup}{m-1} + \frac{1}{2m}\sum_{j = 1}^{m-1} \lambda_j & = 
\frac{\lambdaoup}{m-1} + \frac{1 - \lambdao - \lambdae - \lambdaoup}{2m} =
\frac{m+1}{2m} \left(\frac{\lambdaoup}{m-1} - \frac{\lambdae}{m+1}\right) - \frac{\lambdao}{2m} + \frac{1}{2m}\\
& = \frac{m+1}{2m} \left(\frac{v_2^2}{2 b v_1} - \frac{m}{m+1}\right) - \frac{\lambdao}{2m}
= \frac{v_2+2}{2v_2} \left(\frac{v_2^2}{2 b v_1} - \frac{v_2}{v_2+2}\right) - \frac{\lambdao}{v_2}.
\end{align*}
Thus~\eqref{eq:pr:7:2:2} together with~\eqref{eq:pr:7:2:5} holds if and only if~\eqref{eq:pr:7:2:2} and
\begin{equation*}
\lambdao \geq \left(\frac{v_2}{2b}\right)\left(\frac{v_2(v_2+2)}{2v_1} + b - v_2\right)
\end{equation*}
holds.

Now we show how the conditions on $v_1$, $v_2$, and $b$, follow from the Lemmas~\ref{lemma:7:1}, \ref{lemma:7:2}, and \ref{lemma:7:3}.
By Lemma~\ref{lemma:7:2},
\begin{equation}
\label{eq:pr:7:2:1}
\frac{\lambdae}{m+1} + \frac{\sum_{j = 1}^{m-1} \lambda_j}{2m} \leq \frac{m}{b}\left(1 - \frac{v_2}{v_1}\right).
\end{equation}
Since, by Lemma~\ref{lemma:7:2}, $\lambdae \geq 0$, and $\lambda_1 \geq 0$, ..., $\lambda_{m-1} \geq 0$, so 
\begin{equation*}
\frac{m}{b}\left(1 - \frac{v_2}{v_1}\right) = \frac{m}{b}\left(\frac{v_1 - v_2}{v_1}\right) \geq 0
\end{equation*}
Since $m > b > 0$ so it follows that $v_1 \geq v_2$.

For the conditions on $b$, by~\eqref{eq:pr:7:2:3},
\begin{equation*}
b = \left(\frac{2m}{v_1}\right) \left(\frac{1}{\frac{\lambdao}{m} + \frac{\lambdae}{m+1} + \frac{\lambdaoup}{m-1} + \frac{\sum_{j = 1}^{m-1} \lambda_j}{m}}\right).
\end{equation*}
Since $\lambdao,\lambdae,\lambdaoup,\lambda_1,\ldots,\lambda_{m-1} \geq 0$ and $\lambdao + \lambdae + \lambdaoup + sum_{j = 1}^{m-1} \lambda_j = 1$
so 
\begin{equation*}
\frac{\lambdao}{m} + \frac{\lambdae}{m+1} + \frac{\lambdaoup}{m-1} + \frac{\sum_{j = 1}^{m-1} \lambda_j}{m}
\end{equation*}
is maximised when $\lambdaoup = 1$ and it is minimised when $\lambdae = 1$. Hence
\begin{equation*}
b \in \left[\frac{2m(m-1)}{v_1},\frac{2m(m+1)}{v_1}\right] = \left[\frac{v_2(v_2-2)}{2v_1},\frac{v_2(v_2+2)}{2v_1}\right].
\end{equation*}
\end{proof}

\subsection*{The case of $\Ex(X) = m+\alpha$ and $\Ex(Y) = b$ with $m\in \mathbb{Z}_{\geq 1}$, $m > b > 0$, and $\alpha \in (0,1)$}

\begin{proposition}
\label{pr:8:1}
A strategy profile $(X,Y)$ such that $\Ex(X) = m+\alpha$, $\Ex(Y) = b$,  $m\in \mathbb{Z}_{\geq 1}$, $m > b > 0$, and $\alpha \in (0,1)$, is a Nash equilibrium of all pay auction with both players valuations $v_1,v_2 > 0$ if and only if,
\begin{equation*}
\frac{v_1}{2} > \left\lceil \frac{v_2}{2} \right\rceil > \frac{v_2}{2}, \quad m = \left\lceil \frac{v_2}{2} \right\rceil - 1,\quad b = \frac{\left(\left\lceil \frac{v_2}{2} \right\rceil - 1\right) \left\lceil \frac{v_2}{2} \right\rceil}{\frac{v_1}{2}}, \quad \alpha = \frac{\left\lceil \frac{v_2}{2} \right\rceil}{\frac{v_2}{2}} \left(\frac{v_2}{2} - \left\lfloor \frac{v_2}{2} \right\rfloor \right)
\end{equation*}
\begin{equation*}
X \in \conv(\mathcal{U}^{m,\alpha} \cup \mathcal{X}^{m,\alpha}), \quad Y = \left(1 - \frac{b}{m} \right) \uvec{0} + \left(\frac{b}{m}\right) \ue{m}, 
\end{equation*}
\begin{itemize}
\item $\mathcal{U}^{m,\alpha} = \left\{(1-\alpha)\uo{m} + \alpha \uo{m+1}\right\}$,
\end{itemize}
and
\begin{itemize}
\item $\mathcal{X}^{m,\alpha} = \alpha\delta \mathcal{V}^{m} + \left(1 - \alpha\delta \right) \uo{m}$, if $v_2/2 \leq \lceil v_2/2 \rceil - 1/2$,
\item $\mathcal{X}^{m,\alpha} = (1-\alpha)\sigma \mathcal{V}^{m} + \left(1 - (1-\alpha)\sigma \right) \uo{m+1}$, if $v_2/2 > \lceil v_2/2 \rceil - 1/2$,
\end{itemize}
where
\begin{equation*}
\delta = \frac{2\left\lceil\frac{v_2}{2}\right\rceil - 1}{\left\lceil\frac{v_2}{2}\right\rceil},\quad \sigma = \frac{2\left\lceil\frac{v_2}{2}\right\rceil - 1}{\left\lceil\frac{v_2}{2}\right\rceil-1}.
\end{equation*}

Equilibrium payoffs of the players are
\begin{equation*}
P^1(X,Y) = v_1+1 - 2\left\lceil\frac{v_2}{2}\right\rceil\ \textrm{ and }\ P^2(Y,X) = 0.
\end{equation*}

\end{proposition}

\begin{proposition}
\label{pr:8:2}
A strategy profile $(X,Y)$ such that $\Ex(X) = m+\alpha$, $\Ex(Y) = m$, $m\in \mathbb{Z}_{\geq 1}$, and $\alpha \in (0,1)$, is a Nash equilibrium of all pay auction with both players valuations $v_1,v_2 > 0$ if and only if,
\begin{equation*}
\frac{v_1}{2} = \left\lceil \frac{v_2}{2} \right\rceil, \quad m = \left\lceil \frac{v_2}{2} \right\rceil - 1, \quad \alpha \in \left(0,\frac{\left\lceil \frac{v_2}{2} \right\rceil}{\frac{v_2}{2}} \left(\frac{v_2}{2} - \left\lceil \frac{v_2}{2} \right\rceil + 1\right)\right], \quad Y = \ue{m}, 
\end{equation*}
\begin{itemize}
\item if $0 < \alpha \leq \frac{\left\lceil\frac{v_2}{2}\right\rceil}{2\left\lceil\frac{v_2}{2}\right\rceil - 1}$ then
\begin{align*}
X & = \lambdao ((1-\alpha) \uo{m} + \alpha \uo{m+1}) + \lambdae ((1-\alpha) \ue{m} + \alpha \uo{m+1}) + {}\\
  & \qquad \sum_{j = 1}^{m} \lambda_j \left(\alpha\delta \vd{j}{m} + \left(1 - \alpha\delta \right) \uo{m}\right) + {}\\
  & \qquad \sum_{j = 1}^{m} \kappa_j \left(\alpha\delta \vd{j}{m} + \left(1 - \alpha\delta \right) \ue{m}\right),
\end{align*}
where
\begin{equation*}
\delta = \frac{2\left\lceil\frac{v_2}{2}\right\rceil - 1}{\left\lceil\frac{v_2}{2}\right\rceil},
\end{equation*}
$\lambdao, \lambdae, \lambda_1, \ldots, \lambda_m, \kappa_1, \ldots, \kappa_m \geq 0$, $\lambdao + \lambdae + \sum_{j = 1}^{m} \lambda_j + \sum_{j = 1}^{m} \kappa_j = 1$,
and
\begin{equation*}
\lambdae + \sum_{i = 1}^m \kappa_i \frac{1 - \alpha\delta}{1-\alpha} = \frac{\left\lceil \frac{v_2}{2} \right\rceil\left(\frac{v_2}{2} - \left\lceil \frac{v_2}{2} \right\rceil + 1\right)}{\frac{v_2}{2}(1-\alpha)} - \frac{\alpha}{1-\alpha}
\end{equation*}

\item if $\frac{\left\lceil\frac{v_2}{2}\right\rceil}{2\left\lceil\frac{v_2}{2}\right\rceil - 1} < \alpha \leq \frac{\left\lceil \frac{v_2}{2} \right\rceil}{\frac{v_2}{2}} \left(\frac{v_2}{2} - \left\lceil \frac{v_2}{2} \right\rceil + 1\right)$ then
\begin{align*}
X & = \lambdao ((1-\alpha) \uo{m} + \alpha \uo{m+1}) + \lambdae ((1-\alpha) \ue{m} + \alpha \uo{m+1}) + {}\\
  & \qquad \sum_{j = 1}^{m} \lambda_j \left((1-\alpha)\sigma \vd{j}{m} + \left(1 - (1-\alpha)\sigma \right) \uo{m+1}\right),
\end{align*}
where
\begin{equation*}
\sigma = \frac{2\left\lceil\frac{v_2}{2}\right\rceil - 1}{\left\lceil\frac{v_2}{2}\right\rceil-1},
\end{equation*}
$\lambdao, \lambdae, \lambda_1, \ldots, \lambda_m \geq 0$, $\lambdao + \lambdae + \sum_{j = 1}^{m} \lambda_j = 1$, and
\begin{equation*}
\lambdae = \frac{\left\lceil \frac{v_2}{2} \right\rceil\left(\frac{v_2}{2} - \left\lceil \frac{v_2}{2} \right\rceil + 1\right)}{\frac{v_2}{2}(1-\alpha)} - \frac{\alpha}{1-\alpha}
\end{equation*}
\end{itemize}

Equilibrium payoffs of the players are
\begin{equation*}
P^1(X,Y) = 1\ \textrm{ and }\ P^2(Y,X) = 1-\frac{v_2}{v_1} \alpha - \frac{v_1-v_2}{2}.
\end{equation*}

\end{proposition}

\begin{lemma}
\label{lemma:8:1}
If a strategy profile $(X,Y)$ such that $\Ex(X) = m+\alpha$, $\Ex(Y) = b$,  $m\in \mathbb{Z}_{\geq 1}$, and $m \geq b > 0$, is a Nash equilibrium of all pay auction with both players valuations $v_1,v_2 > 0$ then $m(m+1)/b = v_1/2$, $(X,Y)$ is a Nash equilibrium of $\Gamma(m+\alpha,b)$, and
\begin{equation*}
Y = \left(1 - \frac{b}{m} \right) \uvec{0} + \left(\frac{b}{m}\right) \ue{m}.
\end{equation*}

\end{lemma}

\begin{proof}
Since $(X,Y)$ is a Nash equilibrium of all pay auction so, by Proposition~\ref{pr:allpaylotto}, it is a Nash equilibrium of general Lotto game $\Gamma(m+\alpha,b)$.
Hence, by~\cite[Theorem~4]{Hart08},
\begin{equation}
\label{eq:lemma:8:1:1}
Y = \left(1 - \frac{b}{m} \right) \uvec{0} + \left(\frac{b}{m}\right) \ue{m}
\end{equation}
and
\begin{equation}
\label{eq:lemma:8:1:2}
H(X,Y) = 1 - \frac{(1 - \alpha)b}{m} - \frac{\alpha b}{m+1} = 1 - \frac{b}{m} + \frac{\alpha b}{m(m+1)}.
\end{equation}

Given $c \in (1-m-\alpha,m+1-\alpha)$ let
\begin{equation*}
X' = \gamma \uvec{\lfloor m+\alpha+c \rfloor} + (1 - \gamma)\uvec{\lceil m+\alpha+c \rceil},
\end{equation*}
where
\begin{equation*}
\gamma = \begin{cases}
     1, & \textrm{if $m+c+\alpha \in \mathbb{Z}$}\\
     \frac{\lceil m+c+\alpha \rceil - (m+c+\alpha)}{\lceil m+c+\alpha \rceil - \lfloor m+c+\alpha \rfloor}, & \textrm{otherwise}
     \end{cases}
\end{equation*}
so that $\Ex(X') = m + \alpha + c$.
Since $c \in (1-m-\alpha,m+1-\alpha)$, $m \in \mathbb{Z}_{\geq 1}$, and $\alpha \in (0,1)$ so $\lfloor m+c+\alpha \rfloor \geq 1$ and $\lceil m+c+\alpha \rceil \leq 2m+1$.
Hence $\Pr(X' = 0) = 0$ and $\Pr(X' \geq 2m+2) = 0$. Thus, by~\eqref{eq:lemma:8:1:1}, \eqref{eq:hprop:1}, and~\eqref{eq:hprop:2},
\begin{equation}
\label{eq:lemma:8:1:3}
\begin{aligned}
H(X',Y) & = -\left(1 - \frac{b}{m} \right) H(\uvec{0},X') - \left(\frac{b}{m}\right)H(\ue{m},X')\\
        & = -\left(1 - \frac{b}{m} \right) (\Pr(X' = 0) - 1) - \left(\frac{b}{m}\right) \left(1 - \frac{\Ex(X')+1}{m+1}\right)\\
        & = 1 - \frac{b}{m} + \frac{b(\alpha+c)}{m(m+1)}.
\end{aligned}        
\end{equation}

By~\eqref{eq:pandh}, \eqref{eq:lemma:8:1:2}, and~\eqref{eq:lemma:8:1:3}
\begin{align*}
P(X',Y) & = \frac{v_1}{2}\left(1 - \frac{b}{m} + \frac{b(\alpha+c)}{m(m+1)} - \left(\frac{2(m+c+\alpha)}{v_1} - 1\right)\right), \textrm{ and}\\
P(X,Y) & = \frac{v_1}{2}\left(1 - \frac{b}{m} + \frac{\alpha b}{m(m+1)} - \left(\frac{2(m+\alpha)}{v_1} - 1\right)\right).
\end{align*}
Since $(X,Y)$ is a Nash equilibrium so $P(X,Y) \geq P(X',Y)$. Hence
\begin{multline*}
\frac{v_1}{2}\left(1 - \frac{b}{m} + \frac{\alpha b}{m(m+1)} - \left(\frac{2(m+\alpha)}{v_1} - 1\right)\right) \geq\\
 \frac{v_1}{2}\left(1 - \frac{b}{m} + \frac{b(\alpha+c)}{m(m+1)} - \left(\frac{2(m+c+\alpha)}{v_1} - 1\right)\right)
\end{multline*}
and, consequently
\begin{equation}
\label{eq:lemma:8:1:4}
\frac{2c}{v_1} \geq \frac{bc}{m(m+1)}.
\end{equation}
Since $v_1 > 0$, $b > 0$, $m > 0$, $\alpha \in (0,1)$, and~\eqref{eq:lemma:8:1:4} holds for any $c \in (1-m-\alpha,m+1-\alpha)$ so
$v_1/2 = m(m+1)/b$.
\end{proof}

\begin{lemma}
\label{lemma:8:2}
If a strategy profile $(X,Y)$ such that $\Ex(X) = m+\alpha$, $\Ex(Y) = b$,  $m\in \mathbb{Z}_{\geq 1}$, and $m \geq b > 0$, is a Nash equilibrium of all pay auction with both players valuations $v_1,v_2 > 0$ then
\begin{equation*}
Y = \left(1 - \frac{b}{m} \right) \uvec{0} + \left(\frac{b}{m}\right) \ue{m},
\end{equation*}

\begin{itemize}
\item if $0 < \alpha \leq (m+1)/(2m+1)$ then
\begin{align*}
X & = \lambdao ((1-\alpha) \uo{m} + \alpha \uo{m+1}) + \lambdae ((1-\alpha) \ue{m} + \alpha \uo{m+1}) + {}\\
  & \qquad \sum_{j = 1}^{m} \lambda_j \left(\alpha\delta \vd{j}{m} + \left(1 - \alpha\delta \right) \uo{m}\right) + {}\\
  & \qquad \sum_{j = 1}^{m} \kappa_j \left(\alpha\delta \vd{j}{m} + \left(1 - \alpha\delta \right) \ue{m}\right),
\end{align*}
where $\delta = (2m+1)/(m+1)$,

\item if $(m+1)/(2m+1) < \alpha < 1$ then
\begin{align*}
X & = \lambdao ((1-\alpha) \uo{m} + \alpha \uo{m+1}) + \lambdae ((1-\alpha) \ue{m} + \alpha \uo{m+1}) + {}\\
  & \qquad \sum_{j = 1}^{m} \lambda_j \left((1-\alpha)\sigma \vd{j}{m} + \left(1 - (1-\alpha)\sigma \right) \uo{m+1}\right),
\end{align*}
where $\sigma = (2m+1)/m$.

\item $\lambdao, \lambdae, \lambda_1, \ldots, \lambda_m, \kappa_1, \ldots, \kappa_m \geq 0$, $\lambdao + \lambdae + \sum_{j = 1}^{m} \lambda_j + \sum_{j = 1}^{m} \kappa_j = 1$,
$\lambdae = 0$ if $m > b$, and $\sum_{j = 1}^{m} \kappa_j = 0$ if $m > b$ or $(m+1)/(2m+1) < \alpha < 1$.

\item 
\begin{equation*}
\frac{m+1-\alpha}{m(m+1)} - \frac{2}{v_2} = \lambdae\frac{1-\alpha}{m(m+1)} + \sum_{j = 1}^{m} \kappa_j \left(\frac{1}{m+1}\right)\left(\frac{1-\alpha}{m} - \frac{\alpha}{m+1}\right),
\end{equation*}

\end{itemize}

\end{lemma}

\begin{proof}
Take any strategy profile $(X,Y)$ such that $\Ex(X) = m+\alpha$, $\Ex(Y) = b$,  $m\in \mathbb{Z}_{\geq 1}$, and $m \geq b > 0$, and suppose that it is a Nash equilibrium of all pay auction with both players valuations $v_1,v_2 > 0$.

Since $(X,Y)$ is a Nash equilibrium of all pay auction so, by Proposition~\ref{pr:allpaylotto}, it is a Nash equilibrium of general Lotto game $\Gamma(m+\alpha,b)$.
Thus, by~\cite[Theorem~4]{Hart08} and~\cite[Theorem~6]{Dziubinski12},
\begin{equation}
\label{eq:lemma:8:2:1}
H(X,Y) = 1 - \frac{(1 - \alpha)b}{m} - \frac{\alpha b}{m+1} = 1 - \frac{b}{m} + \frac{\alpha b}{m(m+1)},
\end{equation}
\begin{equation}
\label{eq:lemma:8:2:x}
X \in \conv(\mathcal{U}^{m,\alpha} \cup \mathcal{X}^{m,\alpha})
\end{equation}
and
\begin{equation}
\label{eq:lemma:8:2:y}
Y = \left(1 - \frac{b}{m} \right) \uvec{0} + \left(\frac{b}{m}\right) \ue{m},
\end{equation}
where
\begin{itemize}
\item $\mathcal{U}^{m,\alpha} = (1-\alpha)\mathcal{U}^m + \alpha \uo{m+1}$, if $b = m$,
\item $\mathcal{U}^{m,\alpha} = \left\{(1-\alpha)\uo{m} + \alpha \uo{m+1}\right\}$, if $b < m$
\end{itemize}
and
\begin{itemize}
\item $\mathcal{X}^{m,\alpha} = \alpha\delta \mathcal{V}^{m} + \left(1 - \alpha\delta \right) \mathcal{U}^{m}$, if $0 < \alpha \leq \frac{m+1}{2m+1}$ and $b = m$,
\item $\mathcal{X}^{m,\alpha} = \alpha\delta \mathcal{V}^{m} + \left(1 - \alpha\delta \right) \uo{m}$, if $0 < \alpha \leq \frac{m+1}{2m+1}$ and $b < m$,
\item $\mathcal{X}^{m,\alpha} = (1-\alpha)\sigma \mathcal{V}^{m} + \left(1 - (1-\alpha)\sigma \right) \uo{m+1}$, if $\frac{m+1}{2m+1} < \alpha < 1$, where
\end{itemize}
\begin{equation*}
\delta = \frac{2m+1}{m+1},\quad \sigma = \frac{2m+1}{m}.
\end{equation*}

If $0 < \alpha \leq (m+1)/(2m+1)$ then, by~\eqref{eq:lemma:8:2:x},
\begin{align*}
X & = \lambdao ((1-\alpha) \uo{m} + \alpha \uo{m+1}) + \lambdae ((1-\alpha) \ue{m} + \alpha \uo{m+1}) + {}\\
  & \qquad \sum_{j = 1}^{m} \lambda_j \left(\alpha\delta \vd{j}{m} + \left(1 - \alpha\delta \right) \uo{m}\right) + {}\\
  & \qquad \sum_{j = 1}^{m} \kappa_j \left(\alpha\delta \vd{j}{m} + \left(1 - \alpha\delta \right) \ue{m}\right)
\end{align*}
where $\lambdao, \lambdae, \lambda_1, \ldots, \lambda_m, \kappa_1, \ldots, \kappa_m \geq 0$, $\lambdao + \lambdae + \sum_{j = 1}^{m} \lambda_j + \sum_{j = 1}^{m} \kappa_j = 1$,
and $\lambdae + \sum_{j = 1}^{m} \kappa_j = 0$ if $m > b$.
By~\eqref{eq:hprop:2}, \eqref{eq:uouo}, and \eqref{eq:uov2}, for any strategy $Y'$,
\begin{equation}
\label{eq:lemma:8:2:2}
\begin{aligned}
H(X,Y') & = \lambdao H((1-\alpha) \uo{m} + \alpha \uo{m+1}, Y') + \lambdae H((1-\alpha) \ue{m} + \alpha \uo{m+1}, Y') +  {}\\
& \qquad \sum_{j = 1}^{m} \lambda_j H(\alpha\delta \vd{j}{m} + \left(1 - \alpha\delta \right) \uo{m},Y') + {}\\
& \qquad \sum_{j = 1}^{m} \kappa_j H(\alpha\delta \vd{j}{m} + \left(1 - \alpha\delta \right) \ue{m},Y') + {}\\
        & \geq \lambdao \left(1 - \Ex(Y')\left(\frac{1}{m+1} + \frac{1-\alpha}{m(m+1)}\right)\right) \\
        & \quad {} + \lambdae \left( 1 - \Ex(Y')\left(\frac{1}{m+1} + \frac{1-\alpha}{m(m+1)}\right) + (1-\alpha)\frac{\Ex(Y') - m}{m(m+1)}\right) \\
        & \quad {} + \sum_{j = 1}^{m} \lambda_j \left(1 - \Ex(Y')\left(\frac{1}{m+1} + \frac{1-\alpha}{m(m+1)}\right) + \frac{\alpha}{m+1} \Pr(Y' = 2j-1)\right)\\
        & \quad {} + \sum_{j = 1}^{m} \kappa_j \Bigg(1 - \Ex(Y')\left(\frac{1}{m+1} + \frac{1-\alpha}{m(m+1)}\right) + \frac{\alpha}{m+1} \Pr(Y' = 2j-1)\\
        & \hspace{6cm} {} + \frac{\Ex(Y')-m}{m+1}\left(\frac{1-\alpha}{m} - \frac{\alpha}{m+1}\right)\Bigg)\\
        & = 1 - \Ex(Y')\left(\frac{1}{m+1} + \frac{1-\alpha}{m(m+1)}\right) + \frac{\alpha}{m+1} \sum_{j = 1}^{m} \left(\lambda_j + \kappa_j\right) \Pr(Y' = 2j-1) \\
        & \qquad {} + \lambdae(1-\alpha)\frac{\Ex(Y') - m}{m(m+1)} + \sum_{j = 1}^{m} \kappa_j \left(
                      \frac{\Ex(Y')-m}{m+1}\right)\left(\frac{1-\alpha}{m} - \frac{\alpha}{m+1}\right)
\end{aligned}
\end{equation}
with equality only if $\sum_{j = 2m+1}^{+\infty} \Pr(Y' \geq j) = 0$.

Similarly, if $(m+1)/(2m+1) < \alpha < 1$ then, by~\eqref{eq:lemma:8:2:x},
\begin{equation*}
X = \lambdao ((1-\alpha) \uo{m} + \alpha \uo{m+1}) + \lambdae ((1-\alpha) \ue{m} + \alpha \uo{m+1}) + \sum_{j = 1}^{m} \lambda_j \left((1-\alpha)\sigma \vd{j}{m} + \left(1 - (1-\alpha)\sigma \right) \uo{m+1}\right),
\end{equation*}
where $\lambdao, \lambdae, \lambda_1, \ldots, \lambda_m \geq 0$, $\lambdao + \lambdae + \sum_{j = 1}^{m} \lambda_j = 1$, and $\lambdae = 0$ if $m > b$.
By~\eqref{eq:hprop:2}, \eqref{eq:uouo}, and \eqref{eq:uov2},
\begin{equation}
\label{eq:lemma:8:2:4}
\begin{aligned}
H(X,Y') & = \lambdao H((1-\alpha) \uo{m} + \alpha \uo{m+1}, Y') + \lambdae H((1-\alpha) \ue{m} + \alpha \uo{m+1}, Y') +  {}\\
& \qquad \sum_{j = 1}^{m} \lambda_j H((1-\alpha)\sigma \vd{j}{m} + \left(1 - (1-\alpha)\sigma \right) \uo{m+1},Y')\\
        & \geq \lambdao \left(1 - \Ex(Y')\left(\frac{1}{m+1} + \frac{1-\alpha}{m(m+1)}\right)\right) \\
        & \qquad {} + \lambdae \left( 1 - \Ex(Y')\left(\frac{1}{m+1} + \frac{1-\alpha}{m(m+1)}\right) + (1-\alpha)\frac{\Ex(Y') - m}{m(m+1)}\right) \\
        & \qquad {} + \sum_{j = 1}^{m} \lambda_j \left(1 - \Ex(Y')\left(\frac{1}{m+1} + \frac{1-\alpha}{m(m+1)}\right) + \frac{1-\alpha}{m} \Pr(Y' = 2j-1)\right)\\
        & = 1 - \Ex(Y')\left(\frac{1}{m+1} + \frac{1-\alpha}{m(m+1)}\right) + \frac{1-\alpha}{m}\sum_{j = 1}^{m} \lambda_j \Pr(Y' = 2j-1)\\
        & \qquad {} + \lambdae(1-\alpha)\frac{\Ex(Y') - m}{m(m+1)}
\end{aligned}
\end{equation}
with equality only if $\sum_{j = 2m+1}^{+\infty} \Pr(Y' \geq j) = 0$.

Using~\eqref{eq:lemma:8:2:2}, \eqref{eq:lemma:8:2:4}, \eqref{eq:pandh}, and~\eqref{eq:hprop:1},
\begin{equation}
\label{eq:lemma:8:2:3}
\begin{aligned}
P(Y',X) & = \frac{v_2}{2}\left(H(Y',X) - \left(\frac{2\Ex(Y')}{v_2} - 1\right)\right)
          = \frac{v_2}{2}\left(1 - H(X,Y') - \frac{2\Ex(Y')}{v_2}\right) \\
        & \leq \frac{v_2}{2}\Bigg(\Ex(Y')\left(\frac{1}{m+1} + \frac{1-\alpha}{m(m+1)}\right) - \frac{\alpha}{m+1} \sum_{j = 1}^{m} \left(\lambda_j + \kappa_j\right) \Pr(Y' = 2j-1) \\
        & \qquad {} - \lambdae(1-\alpha)\frac{\Ex(Y') - m}{m(m+1)} - \sum_{j = 1}^{m} \kappa_j \left(
                      \frac{\Ex(Y')-m}{m+1}\right)\left(\frac{1-\alpha}{m} - \frac{\alpha}{m+1}\right) - \frac{2\Ex(Y')}{v_2}\Bigg)\\
        & = \frac{v_2}{2} \Ex(Y')\left(\frac{m+1-\alpha}{m(m+1)} - \frac{2}{v_2}\right) 
            - \frac{v_2}{2} \Bigg(\frac{\alpha}{m+1} \sum_{j = 1}^{m} \left(\lambda_j + \kappa_j\right) \Pr(Y' = 2j-1)\\
        & \qquad\qquad\qquad {} + \lambdae(1-\alpha)\frac{\Ex(Y') - m}{m(m+1)} + \sum_{j = 1}^{m} \kappa_j \left(\frac{\Ex(Y')-m}{m+1}\right)\left(\frac{1-\alpha}{m} - \frac{\alpha}{m+1}\right)\Bigg),
\end{aligned}
\end{equation}
where $\lambdae = 0$ if $m > b$ and $\sum_{j = 1}^{m} \kappa_j = 0$ if either $m > 0$ or $(m+1)/(2m+1) < \alpha < 1$.
\eqref{eq:lemma:8:2:3} holds with equality only if $\sum_{j = 2m+1}^{+\infty} \Pr(Y' \geq j) = 0$.

On the other hand,
\begin{equation}
\label{eq:lemma:8:2:6}
P(Y,X) = \frac{v_2}{2}\left(\frac{\alpha b}{m+1} + \frac{(1-\alpha) b}{m} - 1 - \left(\frac{2b}{v_2} - 1\right)\right) = \frac{v_2}{2}b\left(\frac{m+1-\alpha}{m(m+1)} - \frac{2}{v_2}\right).
\end{equation}

For the remaining part of the proof we consider the cases of $m = b$ and $m > b$ separately.

Assume that $m = b$.
Consider strategy $Y' = ((j-\alpha)/(m+j)) \uvec{0} + ((m+\alpha)/(m+j)) \uvec{m+j}$ with $j = 1$, if $m \bmod 2 = 1$, and $j = 2$, if $m \bmod 2 = 0$. Notice that $\Ex(Y') = m+\alpha$ and $\sum_{j = 2m+1}^{+\infty} \Pr(Y' \geq j) = 0$. The latter follows because if $m \geq 1$ then $m + j \leq 2m$. 
By~\eqref{eq:lemma:8:2:3},
\begin{align*}
P(Y',X) & = \frac{v_2}{2} (m+\alpha)\left(\frac{m+1-\alpha}{m(m+1)} - \frac{2}{v_2}\right) \\
        & \qquad {} - \frac{v_2}{2} \Bigg(\lambdae(1-\alpha)\frac{\alpha}{m(m+1)} + \sum_{j = 1}^{m} \kappa_j \left(\frac{\alpha}{m+1}\right)\left(\frac{1-\alpha}{m} - \frac{\alpha}{m+1}\right)\Bigg).
\end{align*}
On the other hand, since $m = b$ so, by~\eqref{eq:lemma:8:2:6},
\begin{equation}
P(Y,X) = \frac{v_2}{2}m \left(\frac{m+1-\alpha}{m(m+1)} - \frac{2}{v_2}\right).
\end{equation}
Since $(X,Y)$ is a Nash equilibrium so $P(Y,X) \geq P(Y',X)$. Hence it follows that
\begin{equation}
\label{eq:lemma:8:2:11}
\lambdae\frac{1-\alpha}{m(m+1)} + \sum_{j = 1}^{m} \kappa_j \left(\frac{1}{m+1}\right)\left(\frac{1-\alpha}{m} - \frac{\alpha}{m+1}\right) \geq \frac{m+1-\alpha}{m(m+1)} - \frac{2}{v_2}.
\end{equation}
Consider strategy $Y' = \uvec{0}$. If $0 < \alpha \leq (m+1)/(2m+1)$ then, by~\eqref{eq:lemma:8:2:3},
\begin{align*}
P(Y',X) & = \frac{v_2}{2} \Bigg(\lambdae\frac{1-\alpha}{m+1} + \sum_{j = 1}^{m} \kappa_j \left(\frac{m}{m+1}\right)\left(\frac{1-\alpha}{m} - \frac{\alpha}{m+1}\right)\Bigg).
\end{align*}
Since $(X,Y)$ is a Nash equilibrium so $P(Y,X) \geq P(Y',X)$. Hence it follows that
\begin{equation}
\label{eq:lemma:8:2:12}
\frac{m+1-\alpha}{m(m+1)} - \frac{2}{v_2} \geq \lambdae\frac{1-\alpha}{m(m+1)} + \sum_{j = 1}^{m} \kappa_j \left(\frac{1}{m+1}\right)\left(\frac{1-\alpha}{m} - \frac{\alpha}{m+1}\right).
\end{equation}
By~\eqref{eq:lemma:8:2:11} and~\eqref{eq:lemma:8:2:12},
\begin{equation}
\label{eq:lemma:8:2:13}
\frac{m+1-\alpha}{m(m+1)} - \frac{2}{v_2} = \lambdae\frac{1-\alpha}{m(m+1)} + \sum_{j = 1}^{m} \kappa_j \left(\frac{1}{m+1}\right)\left(\frac{1-\alpha}{m} - \frac{\alpha}{m+1}\right).
\end{equation}

Assume that $m > b$. 
Consider any strategy $Y' = ((j-m)/j) \uvec{0} + (m/j) \uvec{j}$ with $j \in \{m,\ldots,2m\}$
and $j \bmod 2 = 0$ ($2m$ is an example of such a $j$). Notice that $\Ex(Y') = m$ and $\sum_{j = 2m+1}^{+\infty} \Pr(Y' \geq j) = 0$.
By~\eqref{eq:lemma:8:2:3},
\begin{equation}
\label{eq:lemma:8:2:8}
P(Y',X) = \frac{v_2}{2} m\left(\frac{m+1-\alpha}{m(m+1)} - \frac{2}{v_2}\right).
\end{equation}
Since $(X,Y)$ is a Nash equilibrium so $P(Y,X) \geq P(Y',X)$ and, by~\eqref{eq:lemma:8:2:6} and~\eqref{eq:lemma:8:2:8}, it follows that
\begin{equation}
\label{eq:lemma:8:2:9}
\frac{v_2}{2}b \left(\frac{m+1-\alpha}{m(m+1)} - \frac{2}{v_2}\right) \geq \frac{v_2}{2} m\left(\frac{m+1-\alpha}{m(m+1)} - \frac{2}{v_2}\right).
\end{equation}
Since $m > b$ so~\eqref{eq:lemma:8:2:9} is satisfied only if
\begin{equation}
\label{eq:lemma:8:2:10}
\frac{m+1-\alpha}{m(m+1)} \leq \frac{2}{v_2}.
\end{equation}
Consider any strategy $Y' = \uvec{0}$. By~\eqref{eq:lemma:8:2:3}, $P(Y',X) = 0$.
Since $(X,Y)$ is a Nash equilibrium so $P(Y,X) \geq P(Y',X)$ and, by~\eqref{eq:lemma:8:2:6} and $b > 0$, it follows that
\begin{equation*}
\frac{m+1-\alpha}{m(m+1)} \geq \frac{2}{v_2}.
\end{equation*}
This, together with~\eqref{eq:lemma:8:2:10} yields
\begin{equation*}
\frac{m+1-\alpha}{m(m+1)} - \frac{2}{v_2} = 0
\end{equation*}
which, given that $\lambdae = 0$ and $\sum_{j = 1}^{m} \kappa_j = 0$ if $m > b$ implies~\eqref{eq:lemma:8:2:13}.

\end{proof}

\begin{lemma}
\label{lemma:8:3}
If conditions in Lemma~\ref{lemma:8:1} and~Lemma~\ref{lemma:8:2} are satisfied for a strategy profile $(X,Y)$ such that $\Ex(X) = m+\alpha$, $\Ex(Y) = b$,  $m\in \mathbb{Z}_{\geq 1}$, and $m \geq b> 0$, then $(X,Y)$ is a Nash equilibrium of all pay auction with both players valuations $v_1,v_2 > 0$
\end{lemma}

\begin{proof}
Take any strategy profile $(X,Y)$, as described in the lemma. By Lemma~\ref{lemma:8:1}, $(X,Y)$ is a Nash equilibrium of General Lotto game $\Gamma(m+\alpha,b)$.
Hence, by~\cite[Theorem~4]{Hart08},
\begin{equation}
\label{eq:lemma:8:3:1}
H(X,Y) = 1 - \frac{(1 - \alpha)b}{m} - \frac{\alpha b}{m+1} = 1 - \frac{b}{m} + \frac{\alpha b}{m(m+1)}.
\end{equation}

We first show that there is no profitable deviation for player $1$. By~\eqref{eq:pandh} and~\eqref{eq:lemma:8:3:1},
\begin{equation}
\label{eq:lemma:8:3:2}
P(X,Y) = \frac{v_1}{2}\left(2 - \frac{(1 - \alpha)b}{m} - \frac{\alpha b}{m+1} - \frac{2(m+\alpha)}{v_1}\right).
\end{equation}
Take any strategy $X'$ with $\Ex(X') \geq 0$. By~\ref{eq:lemma:8:1:3}, \eqref{eq:pandh}, and $b \leq m$,
\begin{equation*}
P(X',Y) \leq \frac{v_1}{2}\left(2 - \left(\frac{b}{m}\right) \left(\frac{2m + 1 - \Ex(X')}{m+1}\right) - \frac{2\Ex(X')}{v_1}\right).
\end{equation*}
By Lemma~\eqref{lemma:8:1}, $v_1/2 = m(m+1)/b$ so
\begin{align*}
P(X',Y) & \leq \frac{v_1}{2}\left(2 - \left(\frac{b}{m}\right) \left(\frac{2m + 1 - \Ex(X')}{m+1}\right) - \frac{b\Ex(X')}{m(m+1)}\right)\\
        & = \frac{v_1}{2}\left(2 - b\left(\frac{2m+1}{m(m+1)}\right)\right)  
\end{align*}
and, by~\eqref{eq:lemma:8:3:2},
\begin{align*}
P(X,Y) & = \frac{v_1}{2}\left(2 - \frac{(1 - \alpha)b}{m} - \frac{\alpha b}{m+1} - \frac{b(m+\alpha)}{m(m+1)}\right) \\
       & = \frac{v_1}{2}\left(2 - b\left(\frac{2m+1}{m(m+1)}\right)\right).
\end{align*}
Hence $P(X,Y) \geq P(X',Y)$. Thus player $1$ has no profitable deviation from $(X,Y)$.

Second, we show that there is no profitable deviation for player $2$. Take any strategy $Y'$ with $\Ex(Y') \geq 0$. 
By~\eqref{eq:lemma:8:2:3},
\begin{align*}
P(Y',X) & \leq \frac{v_2}{2} \Ex(Y')\left(\frac{m+1-\alpha}{m(m+1)} - \frac{2}{v_2}\right) \\
        & \qquad {} - \frac{v_2}{2} (\Ex(Y') - m) \Bigg(\lambdae\frac{1-\alpha}{m(m+1)} + \sum_{j = 1}^{m} \kappa_j \left(\frac{1}{m+1}\right)\left(\frac{1-\alpha}{m} - \frac{\alpha}{m+1}\right)\Bigg).
\end{align*}
By Lemma~\ref{lemma:8:2},
\begin{equation}
\label{eq:lemma:8:3:3}
\frac{m+1-\alpha}{m(m+1)} - \frac{2}{v_2} = \lambdae\frac{1-\alpha}{m(m+1)} + \sum_{j = 1}^{m} \kappa_j \left(\frac{1}{m+1}\right)\left(\frac{1-\alpha}{m} - \frac{\alpha}{m+1}\right).
\end{equation}
Hence
\begin{align*}
P(Y',X) & \leq \frac{v_2}{2} m\left(\frac{m+1-\alpha}{m(m+1)} - \frac{2}{v_2}\right).
\end{align*}
On the other hand, by~\eqref{eq:lemma:8:2:6},
\begin{equation}
P(Y,X) = \frac{v_2}{2}b\left(\frac{m+1-\alpha}{m+1} - \frac{2}{v_2}\right).
\end{equation}
If $m = b$ then $P(Y,X) = P(Y',X)$. If $m > b$ then, by Lemma~\ref{lemma:8:2}, $\lambdae = \sum_{j = 1}^{m} \kappa_j = 0$ and, by~\eqref{eq:lemma:8:3:3}, 
\begin{equation*}
\frac{m+1-\alpha}{m(m+1)} - \frac{2}{v_2} = 0.
\end{equation*}
Hence in this case $P(Y,X) = P(Y',X)$ as well. Thus there is no profitable deviation for player $2$.
\end{proof}

\begin{proof}[Proof of Proposition~\ref{pr:8:1}]
By Lemmas~\ref{lemma:8:1}, \ref{lemma:8:2}, and~\ref{lemma:8:3}, a strategy profile $(X,Y)$ with $\Ex(X) = m+\alpha$, $\Ex(Y) = b$,  $m\in \mathbb{Z}_{\geq 1}$, and $m > b > 0$,
is a Nash equilibrium of all pay auction with both players valuations $v_1,v_2 > 0$ if and only if the conditions stated in the lemmas are satisfied.

Thus
\begin{equation}
\label{eq:pr:8:1}
\frac{m+1-\alpha}{m(m+1)} - \frac{2}{v_2} = \lambdae\frac{1-\alpha}{m(m+1)} + \sum_{j = 1}^{m} \kappa_j \left(\frac{1}{m+1}\right)\left(\frac{1-\alpha}{m} - \frac{\alpha}{m+1}\right)
\end{equation}
with $\lambdae = 0$ and $\sum_{j = 1}^{m} \kappa_j =0$, as $m > b$. Thus
\begin{equation*}
\frac{m+1-\alpha}{m(m+1)} = \frac{2}{v_2}
\end{equation*}
and, since $\alpha \in (0,1)$ and $m > b$ so $v_2 = 2m(m+1)/(m+1-\alpha) < 2m(m+1)/b = v_1$ and
\begin{equation*}
m < \frac{v_2}{2} < m+1.
\end{equation*}
Since $m \in \mathbb{Z}_{\geq 1}$ so it follows that $m = \lceil v_2/2 \rceil - 1$ and $ \lceil v_2/2 \rceil > v_2/2$.
By conditions of Lemma~\ref{lemma:8:1}, $m(m+1)/b = v_1/2$. Hence
\begin{equation*}
b = \frac{\left(\left\lceil \frac{v_2}{2} \right\rceil - 1\right) \left\lceil \frac{v_2}{2} \right\rceil}{\frac{v_1}{2}}.
\end{equation*}
Since $m > b$ so $v_1/2 = m(m+1)/b > m+1 > \lceil v_2/2 \rceil$. Moreover
\begin{equation*}
\alpha = \frac{(m+1)\left(\frac{v_2}{2}-m\right)}{\frac{v_2}{2}} = \frac{\left\lceil \frac{v_2}{2} \right\rceil \left(\frac{v_2}{2} - \left\lfloor \frac{v_2}{2} \right\rfloor\right)}{\frac{v_2}{2}}
\end{equation*}

By Lemma~\ref{lemma:8:2} the characterization of $X$ and $Y$ is as stated in the theorem (the cases in the definition of $\mathcal{X}^{\alpha,m}$ follow immediately by substituting the formulas for $m$ and $\alpha$).
\end{proof}

\begin{proof}[Proof of Proposition~\ref{pr:8:2}]
By Lemmas~\ref{lemma:8:1}, \ref{lemma:8:2}, and~\ref{lemma:8:3}, a strategy profile $(X,Y)$ with $\Ex(X) = m+\alpha$, $\Ex(Y) = m$,  $m\in \mathbb{Z}_{\geq 1}$,
is a Nash equilibrium of all pay auction with both players valuations $v_1,v_2 > 0$ if and only if the conditions stated in the lemmas are satisfied.

Thus
\begin{equation*}
\frac{m+1-\alpha}{m(m+1)} - \frac{2}{v_2} = \lambdae\frac{1-\alpha}{m(m+1)} + \sum_{j = 1}^{m} \kappa_j \left(\frac{1}{m+1}\right)\left(\frac{1-\alpha}{m} - \frac{\alpha}{m+1}\right).
\end{equation*}
Since $\lambdae + \sum_{j = 1}^{m} \kappa_j \leq 1$ and
\begin{equation*}
\frac{1-\alpha}{m(m+1)} > \left(\frac{1}{m+1}\right)\left(\frac{1-\alpha}{m} - \frac{\alpha}{m+1}\right)
\end{equation*}
so the RHS of~\eqref{eq:pr:8:1} is maximised when $\lambdae = 1$ and $\sum_{j = 1}^{m} \kappa_j = 0$ and it is minimised when $\lambdae = \sum_{j = 1}^{m} \kappa_j = 0$.
Therefore
\begin{equation*}
0 \leq \frac{m+1-\alpha}{m(m+1)} - \frac{2}{v_2} \leq \frac{1-\alpha}{m(m+1)}
\end{equation*}
and it follows that
\begin{equation*}
\frac{m(m+1)}{m+1-\alpha} \leq \frac{v_2}{2} \leq m+1.
\end{equation*}
Since $\alpha \in (0,1)$ so if further follows that 
\begin{equation*}
m < \frac{v_2}{2} \leq m+1.
\end{equation*}
Thus $m = \lceil v_2/2 \rceil - 1$.
By Lemma~\ref{lemma:8:1}, $v_1/2 = m(m+1)/m = m+1 = \lceil v_2/2 \rceil$.

By Lemma~\ref{lemma:8:2} the characterization of $X$ and $Y$ is as stated in the theorem (the cases in the definition of $X$ follow immediately by substituting the formulas for $m$ and $\alpha$).

By Lemma~\ref{lemma:8:2},
\begin{equation}
\label{eq:pr:8:2:1}
\frac{m+1-\alpha}{m(m+1)} - \frac{2}{v_2} = \lambdae\frac{1-\alpha}{m(m+1)} + \sum_{j = 1}^{m} \kappa_j \left(\frac{1}{m+1}\right)\left(\frac{1-\alpha}{m} - \frac{\alpha}{m+1}\right),
\end{equation}
with $\lambdao, \lambdae, \lambda_1, \ldots, \lambda_m, \kappa_1, \ldots, \kappa_m \geq 0$, $\lambdao + \lambdae + \sum_{j = 1}^{m} \lambda_j + \sum_{j = 1}^{m} \kappa_j = 1$,
and $\sum_{j = 1}^{m} \kappa_j = 0$ if $(m+1)/(2m+1) < \alpha < 1$. 
Substituting $\lceil v_2/2 \rceil - 1$ for $m$ and reorganizing we obtain
\begin{equation*}
\lambdae + \sum_{i = 1}^m \kappa_i \frac{1 - \alpha\delta}{1-\alpha} = \frac{\left\lceil \frac{v_2}{2} \right\rceil\left(\frac{v_2}{2} - \left\lceil \frac{v_2}{2} \right\rceil + 1\right)}{\frac{v_2}{2}(1-\alpha)} - \frac{\alpha}{1-\alpha}
\end{equation*}
and the conditions for $\lambdae$ and $\sum_{i = 1}^m \kappa_i$ stated in the theorem follow.

For the bounds on $\alpha$, by~\eqref{eq:pr:8:2:1},
\begin{equation*}
\alpha = \frac{\frac{(v_2 - 2m)(m+1)}{v_2} - \lambdae - \sum_{j = 1}^{m} \kappa_j}{1 - \lambdae - \delta\sum_{j = 1}^{m} \kappa_j}
       = 1 - \frac{(1-\delta)\sum_{j = 1}^{m} \kappa_j + \frac{m(2(m+1)-v_2)}{v_2}}{1 - \lambdae - \delta\sum_{j = 1}^{m} \kappa_j}.
\end{equation*}
Since $\delta < 1$, $v_2/2 \leq (m+1)$, and $\lambdae + \sum_{j = 1}^{m} \kappa_j \leq 1$ so it follows that the RHS of the equality is decreasing when $\lambdae$ increases.
Similarly, 
\begin{equation*}
\alpha = \frac{\frac{(v_2 - 2m)(m+1)}{v_2} - \lambdae - \sum_{j = 1}^{m} \kappa_j}{1 - \lambdae - \delta\sum_{j = 1}^{m} \kappa_j}
       = \frac{1}{\delta}\left(1 - \frac{(1 -\delta)(1-\lambdae) + \delta\left(\frac{m(2(m+1)-v_2)}{v_2}\right)}{1 - \lambdae - \delta\sum_{j = 1}^{m} \kappa_j 	}\right).
\end{equation*}
and the RHS of the equality is decreasing when $\lambdae< 0$ and $\sum_{j = 1}^{m} \kappa_j$ increases. Thus the RHS attains its maximal value when $\lambdae = \sum_{j = 1}^{m} \kappa_j = 0$. The RHS obtains it minimal value, $0$, when $\sum_{j = 1}^m \kappa_j = 0$ and
\begin{equation*}
\lambdae = 1 - \frac{m(2(m+1)-v_2)}{v_2} \geq 1 - \left(m+1 - \frac{v_2}{2}\right) = 1 - \left(\left\lceil\frac{v_2}{2}\right\rceil - \frac{v_2}{2}\right).
\end{equation*}
Thus it follows that
\begin{equation*}
0 < \alpha \leq \frac{(m+1)(v_2 - 2m)}{v_2} = \frac{\left\lceil \frac{v_2}{2} \right\rceil}{\frac{v_2}{2}} \left(\frac{v_2}{2} - \left\lceil \frac{v_2}{2} \right\rceil + 1\right).
\end{equation*}

\end{proof}

\subsection*{Proofs of Theorems~\ref{th:int} and~\ref{th:nonint}}

In this section we use the characterisations of equilibrium strategy profiles for all the possible configurations of expected values of equilibrium strategies, 
obtained in Propositions~\ref{pr:5}, \ref{pr:6}, \ref{pr:7:1}, \ref{pr:7:2}, \ref{pr:8:1}, and~\ref{pr:8:2}, to give proofs of the two Theorems~\ref{th:int}
and~\ref{th:nonint}.

\begin{proof}[Proof of Theorem~\ref{th:int}]

\textbf{Point~\eqref{th:int:1}.}\\
\noindent For point~\eqref{th:int:1} assume that $v_2/2 \in \mathbb{Z}$, $v_1 = v_2 > 0$ and let $m = v_2/2 - 1$.

For the right to left implication take any strategy profile $(X,Y)$ with $X = \alpha \uo{m+1} + (1-\alpha) \ue{m}$
and $Y = \beta \uo{m+1} + (1-\beta) \ue{m}$.
If $\alpha = \beta = 0$ then $\Ex(X) = \Ex(Y) = m$ and $(X,Y)$ is a Nash equilibrium by Proposition~\ref{pr:5}.
If $\alpha = \beta = 1$ then $\Ex(X) = \Ex(Y) = v_2/2$ and again $(X,Y)$ is a Nash equilibrium by Proposition~\ref{pr:5}.

If $\alpha \in (0,1)$ and $\beta \in (0,1)$ then $\Ex(X) = m+ \alpha$, $\Ex(Y) = m+ \beta$ and $(X,Y)$ is a Nash equilibrium by Proposition~\ref{pr:6}.

In the case of $\alpha = 1$ and $\beta \in (0,1)$ let $m' = v_2/2$. Then $\Ex(X) = m'$ and $\Ex(Y) = m'-(1-\beta)$, $X = \uo{m'}$,
and $Y = \beta \uo{m'} + (1-\beta)\ue{m'-1}$.

Suppose that $v_2/2 = 1$ and let $b = \beta$. Then $m' = 1$, $\Ex(X) = 1$, $\Ex(Y) = b$, and $\ue{m'-1} = \uvec{0}$. 
Thus $Y = b \uo{m'} + (1-b)\uvec{0} = (1-b)\uvec{0} + b(\lambda \uo{m'} + (1-\lambda)\ue{m'})$
with $\lambda = 1$. Since (by $v_1 = v_2 = 2$)
\begin{equation*}
\frac{4}{bv_1} - \frac{2}{b} + 1 = 1 \leq \lambda \leq \frac{2-\beta}{\beta} = \frac{4}{bv_1}-1
\end{equation*}
so, by Proposition~\ref{pr:7:1} $(X,Y)$ is a Nash equilibrium.

Suppose that $v_2/2 \geq 2$ and let $b = m'-(1 - \beta)$.
Since
\begin{equation}
\label{eq:th:int:2}
\ue{m'-1} = \left(1 - \frac{m'-1}{m'}\right)\uvec{0} + \left(\frac{m'-1}{m'}\right)\uoup{m'}
\end{equation}
so
\begin{align*}
Y & = (1-\beta)\left(1 - \frac{m'-1}{m'}\right)\uvec{0} + \beta\uo{m'} + (1-\beta)\left(\frac{m'-1}{m'}\right)\uoup{m'}\\
  & = \left(1 - \frac{b}{m'}\right)\uvec{0} + \left(\frac{b}{m'}\right)\left(\lambda\uo{m'} + (1-\lambda)\uoup{m'}\right)\\
\end{align*}
where $\lambda = \beta m'/b$. Thus
\begin{align*}
Y = \left(1 - \frac{b}{m'}\right)\uvec{0} + \left(\frac{b}{m'}\right)Z
\end{align*}
where
\begin{equation*}
Z = \lambdao\uo{m} + \lambdae\ue{m} + \lambdaoup \uoup{m} + \sum_{j=1}^{m-1} \lambda_j \w{j}{m} 
\end{equation*}
and $\lambdao = \beta m'/b$, $\lambdaoup = 1 - \beta m'/b$, and $\lambdae = \lambda_1 =\ldots=\lambda_{m-1} = 0$.
Since
\begin{equation*}
\frac{\lambdaoup}{m'-1} - \frac{\lambdae}{m'+1} = \frac{b - \beta m'}{b(m'-1)} = \frac{m'(b - \beta)}{b(m'-1)} -1 = \frac{m'}{b} - 1 = \frac{v_2^2}{2v_1b}-1,
\end{equation*}
\begin{equation*}
\frac{\lambdaoup}{m-1} + \frac{1}{2m}\sum_{j = 1}^{m-1} \lambda_j = \frac{m'}{b} - 1 = \frac{v_2}{2b}-1,
\end{equation*}
and
\begin{equation*}
\frac{v_2(v_2-2)}{2v_1} = m'-1 < m'-1+\beta = b < m' < \frac{v_2(v_2 + 2)}{2v_1}
\end{equation*}
so, by Proposition~\ref{pr:7:2}, $(X,Y)$ is a Nash equilibrium.
Analogously, if $\alpha \in (0,1)$ and $\beta = 1$ then $(X,Y)$ is a Nash equilibrium by Propositions~\ref{pr:7:1} and~\ref{pr:7:2}. 

If $\alpha \in (0,1)$ and $\beta = 0$ then $\Ex(X) = m + \alpha$, $\Ex(Y) = m$, and $Y = \ue{m}$.
Suppose that $0 < \alpha \leq (m+1)/(2m+1)$. Then
\begin{align*}
X & = \alpha \uo{m+1} + (1-\alpha)\ue{m}\\
  & = \lambdao ((1-\alpha) \uo{m} + \alpha \uo{m+1}) + \lambdae ((1-\alpha) \ue{m} + \alpha \uo{m+1}) + {}\\
  & \qquad \sum_{j = 1}^{m} \lambda_j \left(\alpha\delta \vd{j}{m} + \left(1 - \alpha\delta \right) \uo{m}\right) + {}\\
  & \qquad \sum_{j = 1}^{m} \kappa_j \left(\alpha\delta \vd{j}{m} + \left(1 - \alpha\delta \right) \ue{m}\right),
\end{align*}
with $\lambdae = 1$  and $\lambdao = \lambda_1 = \ldots = \lambda_m = \kappa_1 = \ldots  = \kappa_m = 0$.
Since 
\begin{equation}
\label{eq:th:int:1}
\frac{\lceil\frac{v_2}{2}\rceil}{\frac{v_2}{2}}\left(\frac{v_2}{2} - \left\lceil\frac{v_2}{2}\right\rceil + 1\right) = 1
\end{equation}
so
\begin{equation*}
\lambdae + \sum_{i = 1}^m \kappa_i \frac{1 - \alpha\delta}{1-\alpha} = 1 =  \frac{\left\lceil \frac{v_2}{2} \right\rceil\left(\frac{v_2}{2} - \left\lceil \frac{v_2}{2} \right\rceil + 1\right)}{\frac{v_2}{2}(1-\alpha)} - \frac{\alpha}{1-\alpha}
\end{equation*}
and so, by Proposition~\ref{pr:8:2} $(X,Y)$ is a Nash equilibrium.
Similarly, if $(m+1)/(2m+1) \leq \alpha < 1$ then
\begin{equation*}
\alpha < \frac{\lceil\frac{v_2}{2}\rceil}{\frac{v_2}{2}}\left(\frac{v_2}{2} - \left\lceil\frac{v_2}{2}\right\rceil + 1\right)
\end{equation*}
and 
\begin{align*}
X & = \alpha \uo{m+1} + (1-\alpha)\ue{m}\\
  & = \lambdao ((1-\alpha) \uo{m} + \alpha \uo{m+1}) + \lambdae ((1-\alpha) \ue{m} + \alpha \uo{m+1}) + {}\\
  & \qquad \sum_{j = 1}^{m} \lambda_j \left((1-\alpha)\sigma \vd{j}{m} + \left(1 - (1-\alpha)\sigma \right) \uo{m+1}\right),
\end{align*}
with $\lambdae = 1$, and $\lambdao = \lambda_1 = \ldots = \lambda_m = 0$.
Since 
\begin{equation*}
\lambdae = 1 = \frac{\left\lceil \frac{v_2}{2} \right\rceil\left(\frac{v_2}{2} - \left\lceil \frac{v_2}{2} \right\rceil + 1\right)}{\frac{v_2}{2}(1-\alpha)} - \frac{\alpha}{1-\alpha}
\end{equation*}
so, by Proposition~\ref{pr:8:2} $(X,Y)$ is a Nash equilibrium.
Analogously, if $\alpha = 0$ and $\beta \in (0,1)$ then $(X,Y)$ is a Nash equilibrium by Proposition~\ref{pr:8:2}.
This completes proof of the right to left implication.

For the left to right implication, suppose that $(X,Y)$ is a Nash equilibrium of all-pay auction with $v_1 = v_2 > 0$ and $v_2/2 \in \mathbb{Z}$.
Then, by Propositions~\ref{pr:5}, \ref{pr:6}, \ref{pr:7:1}, \ref{pr:7:2}, \ref{pr:8:1}, and~\ref{pr:8:2}, one of the following cases holds:
\begin{enumerate}
\item $\Ex(X) = \Ex(Y) = v_2/2-1$,\label{p:1:1}
\item $\Ex(X) = \Ex(Y) = v_2/2$,\label{p:1:2}
\item $\Ex(X) = m+\alpha$ and $\Ex(Y) = m+\beta$, where $m = v_2/2-1$ and $\alpha,\beta\in (0,1)$,\label{p:1:3}
\item $\Ex(X) = 1$ and $\Ex(Y) = b$, $v_2 = 2$, and $b \in (0,1)$.\label{p:1:4:1}
\item $\Ex(Y) = 1$ and $\Ex(X) = b$, $v_2 = 2$, and $b \in (0,1)$.\label{p:1:4:2}
\item $\Ex(X) = m'$ and $\Ex(Y) = b$, where $m' = v_2/2 \geq 2$ and $m' > b > 0$.\label{p:1:5:1}
\item $\Ex(Y) = m'$ and $\Ex(X) = b$, where $m' = v_2/2 \geq 2$ and $m' > b > 0$.\label{p:1:5:2}
\item $\Ex(X) = m+\alpha$ and $\Ex(Y) = m$, where $m = v_2/2-1 > 0$ and $\alpha \in (0,1)$.\label{p:1:6:1}
\item $\Ex(Y) = m+\beta$ and $\Ex(X) = m$, where $m = v_2/2-1 > 0$ and $\beta \in (0,1)$.\label{p:1:6:2}
\end{enumerate}
In case~\ref{p:1:1}, by Proposition~\ref{pr:5}, $X = Y = \ue{m} = \alpha\uo{m+1} + (1-\alpha)\ue{m}$ with $\alpha = 0$.
In case~\ref{p:1:2}, by Proposition~\ref{pr:5}, $X = Y = \uo{m+1} = \alpha \uo{m+1} + (1-\alpha)\ue{m}$ with $\alpha = 1$.
In case~\ref{p:1:3}, by Proposition~\ref{pr:6} $X = \alpha\uo{m+1} + (1-\alpha)\ue{m}$ and $Y = \beta\uo{m+1} + (1-\beta)\ue{m}$.

In case~\ref{p:1:4:1}, $m = v_2/2-1=0$ and, by Proposition~\ref{pr:7:1}, $X = \uo{1} = \alpha \uo{1} + (1-\alpha)\ue{0}$ with $\alpha = 1$,
and $Y = (1-b)\uvec{0} + b(\lambda\uo{1} + (1-\lambda)\ue{1})$ with $\lambda \in [0,1]$ and
\begin{equation*}
\lambda \geq \frac{4}{bv_1} - \frac{2}{b} + 1 = 1.
\end{equation*}
Hence $\lambda = 1$ and $Y = (1-b)\uvec{0} + b\uo{1} = b\uo{1} + (1-b)\ue{0}$ with $b \in (0,1)$.
Analogously, in case~\ref{p:1:4:2}, $Y = \uo{1} = \beta \uo{1} + (1-\beta)\ue{0}$ with $\beta = 1$, and
$X = \alpha\uo{1} + (1-\alpha)\ue{0}$ with $\alpha = b \in (0,1)$.

In case~\ref{p:1:5:1}, $m = v_2/2-1 = m'-1$ and, by Proposition~\ref{pr:7:2}, $X = \uo{m'} = \alpha \uo{m+1} + (1-\alpha)\ue{m}$ with $\alpha = 1$,
and 
\begin{equation*}
Y = \left(1 - \frac{b}{m'}\right)\uvec{0} + \left(\frac{b}{m'}\right)Z
\end{equation*}
where
\begin{equation*}
Z = \lambdao\uo{m'} + \lambdae\ue{m'} + \lambdaoup \uoup{m'} + \sum_{j=1}^{m-1} \lambda_j \w{j}{m'},
\end{equation*}
$0<b<m'$,
\begin{equation*}
\frac{v_2(v_2-2)}{2v_1} \leq b \leq \frac{v_2(v_2+2)}{2v_1},
\end{equation*}
\begin{equation*}
\lambdao,\lambdae,\lambdaoup,\lambda_1,\ldots,\lambda_{m'-1} \geq 0 \textrm{ and } \lambdao + \lambdae + \lambdaoup + \sum_{j = 1}^{m'-1} \lambda_j = 1,
\end{equation*}
\begin{equation*}
\frac{\lambdaoup}{m'-1} - \frac{\lambdae}{m'+1} = \frac{v_2^2}{2v_1b} - 1,
\end{equation*}
and
\begin{equation*}
\frac{\lambdaoup}{m'-1} + \frac{1}{2m'}\sum_{j = 1}^{m'-1} \lambda_j \leq \frac{v_2}{2b}-1.
\end{equation*}
Hence
\begin{equation*}
\lambdaoup = \frac{m'(m'-1)}{b} - (m'-1) + \lambdae\frac{m'-1}{m'+1},
\end{equation*}
\begin{equation*}
\frac{\lambdae}{m'+1} + \frac{1}{2m'}\sum_{j = 1}^{m'-1} \lambda_j \leq \frac{v_2v_1}{2v_1b} - \frac{v_2^2}{2v_1b} = 0,
\end{equation*}
and so $\lambdae = \lambda_1 = \ldots = \lambda_{m'-1} = 0$, $\lambdaoup = 1 - \lambdao$, and
\begin{equation*}
\lambdao = m'\left(1-\frac{m'-1}{b}\right).
\end{equation*}
This, together with~\eqref{eq:th:int:2} yields
\begin{equation*}
Y = \left(1 - \frac{b}{m'}\right)\uvec{0} + \left(\frac{b}{m'}\right)\left(\lambdao\uo{m'} + (1-\lambdao) \uoup{m'}\right)
  = (1-(m'-b))\uo{m'} + (m'-b)\ue{m'-1}.
\end{equation*}
Since 
\begin{equation*}
\frac{v_2(v_2-2)}{2v_1} = \frac{v_2}{2}-1 = m'-1 \leq b < m'
\end{equation*}
so $m'-b \in (0,1]$ and so $Y = \beta\uo{m+1} + (1-\beta)\ue{m}$ with $\beta = m'-b \in (0,1]$.
Analogously, in case~\ref{p:1:5:2}, $Y = \uo{m+1} = \beta \uo{m+1} + (1-\beta)\ue{m}$ with $\beta = 1$, and
$X = \alpha\uo{m+1} + (1-\alpha)\ue{m}$ with $\alpha = m'-b \in (0,1]$.

In case~\ref{p:1:6:1}, by Proposition~\ref{pr:8:2}, $Y = \ue{m} = \beta\uo{m+1} + (1-\beta)\ue{m}$ with $\beta = 0$, and
\begin{align*}
X & = \lambdao ((1-\alpha) \uo{m} + \alpha \uo{m+1}) + \lambdae ((1-\alpha) \ue{m} + \alpha \uo{m+1}) + {}\\
  & \qquad \sum_{j = 1}^{m} \lambda_j \left(\alpha\delta \vd{j}{m} + \left(1 - \alpha\delta \right) \uo{m}\right) + {}\\
  & \qquad \sum_{j = 1}^{m} \kappa_j \left(\alpha\delta \vd{j}{m} + \left(1 - \alpha\delta \right) \ue{m}\right),
\end{align*}
with
\begin{equation*}
\delta = \frac{2\left\lceil\frac{v_2}{2}\right\rceil - 1}{\left\lceil\frac{v_2}{2}\right\rceil},
\end{equation*}
$\lambdao, \lambdae, \lambda_1, \ldots, \lambda_m, \kappa_1, \ldots, \kappa_m \geq 0$, $\lambdao + \lambdae + \sum_{j = 1}^{m} \lambda_j + \sum_{j = 1}^{m} \kappa_j = 1$,
and
\begin{equation*}
\lambdae + \sum_{i = 1}^m \kappa_i \frac{1 - \alpha\delta}{1-\alpha} = \frac{\left\lceil \frac{v_2}{2} \right\rceil\left(\frac{v_2}{2} - \left\lceil \frac{v_2}{2} \right\rceil + 1\right)}{\frac{v_2}{2}(1-\alpha)} - \frac{\alpha}{1-\alpha},
\end{equation*}
in the case of $0 < \alpha \leq \frac{\left\lceil\frac{v_2}{2}\right\rceil}{2\left\lceil\frac{v_2}{2}\right\rceil - 1}$,
and
\begin{align*}
X & = \lambdao ((1-\alpha) \uo{m} + \alpha \uo{m+1}) + \lambdae ((1-\alpha) \ue{m} + \alpha \uo{m+1}) + {}\\
  & \qquad \sum_{j = 1}^{m} \lambda_j \left((1-\alpha)\sigma \vd{j}{m} + \left(1 - (1-\alpha)\sigma \right) \uo{m+1}\right),
\end{align*}
with
\begin{equation*}
\sigma = \frac{2\left\lceil\frac{v_2}{2}\right\rceil - 1}{\left\lceil\frac{v_2}{2}\right\rceil-1},
\end{equation*}
$\lambdao, \lambdae, \lambda_1, \ldots, \lambda_m \geq 0$, $\lambdao + \lambdae + \sum_{j = 1}^{m} \lambda_j = 1$, and
\begin{equation*}
\lambdae = \frac{\left\lceil \frac{v_2}{2} \right\rceil\left(\frac{v_2}{2} - \left\lceil \frac{v_2}{2} \right\rceil + 1\right)}{\frac{v_2}{2}(1-\alpha)} - \frac{\alpha}{1-\alpha},
\end{equation*}
in the case of $\frac{\left\lceil\frac{v_2}{2}\right\rceil}{2\left\lceil\frac{v_2}{2}\right\rceil - 1} < \alpha \leq \frac{\left\lceil \frac{v_2}{2} \right\rceil}{\frac{v_2}{2}} \left(\frac{v_2}{2} - \left\lceil \frac{v_2}{2} \right\rceil + 1\right)$.
If $0 < \alpha \leq \frac{\left\lceil\frac{v_2}{2}\right\rceil}{2\left\lceil\frac{v_2}{2}\right\rceil - 1}$ then, by~\eqref{eq:th:int:1},
\begin{equation*}
\lambdae + \sum_{i = 1}^m \kappa_i \frac{1 - \alpha\delta}{1-\alpha} = 1.
\end{equation*}
Hence $\lambdao = \lambda_1 = \ldots \ldots \lambda_m = 0$.
Since $v_2/2-1 > 0$ so $\delta > 1$ and $(1 - \alpha\delta)/(1-\alpha) < 1$. Hence $\lambdae = 1$, $\kappa_1 = \ldots = \kappa_m = 0$, and
$X = (1-\alpha) \ue{m} + \alpha \uo{m+1}$ with $\alpha \in (0,1)$.
If $\frac{\left\lceil\frac{v_2}{2}\right\rceil}{2\left\lceil\frac{v_2}{2}\right\rceil - 1} < \alpha \leq \frac{\left\lceil \frac{v_2}{2} \right\rceil}{\frac{v_2}{2}} \left(\frac{v_2}{2} - \left\lceil \frac{v_2}{2} \right\rceil + 1\right)$
then, by~\eqref{eq:th:int:1}, $\lambdae = 1$ and, consequently, $\lambdao = \lambda_1 = \ldots \ldots \lambda_m = 0$. Thus $X = (1-\alpha) \ue{m} + \alpha \uo{m+1}$ with $\alpha \in (0,1)$.
Analogously, in case~\ref{p:1:6:2}, $X = \ue{m} = \alpha \uo{m+1} + (1-\alpha)\ue{m}$ with $\alpha = 1$, and
$Y = \beta\uo{m+1} + (1-\beta)\ue{m}$ with $\beta \in (0,1)$.

Thus we have shown that in all the cases~\ref{p:1:1}--\ref{p:1:6:2}, $(X,Y)$ is as stated in point~\eqref{th:int:1} which completes the proof of the left to right implication.

\textbf{Point~\eqref{th:int:2}.}\\
\noindent For point~\eqref{th:int:2} assume that $v_2/2 \in \mathbb{Z}$ and $v_1 > v_2 = 2$.

For the right to left implication take any strategy profile $(X,Y)$ with
$X = \uo{1}$ and $Y = (1 - b) \uvec{0} + b(\lambda \uo{1} + (1-\lambda) \ue{1})$, 
where $b\in (0,1]$ and
\begin{equation*}
\frac{4}{bv_1} - \frac{2}{b} + 1 \leq \lambda \leq \frac{4}{bv_1} - 1.
\end{equation*}

If $b = 1$ then $X = \uo{1} = \alpha \uo{1} + (1-\alpha) \ue{1}$ with
\begin{equation*}
\alpha = 1 = \frac{2m}{v_2}\left(m+1-\frac{v_2}{2}\right)
\end{equation*}
and $Y = \lambda \uo{1} + (1-\lambda)\ue{1}$ with
\begin{equation*}
\frac{2m}{v_1}\left(m+1-\frac{v_1}{2}\right) = \frac{4}{v_1} - 1 =  \frac{4}{bv_1} - \frac{2}{b} + 1 \leq \lambda \leq \frac{4}{bv_1} - 1 = \frac{4}{v_1}-1 = \frac{2m}{v_1}\left(m+1-\frac{v_1}{2}\right),
\end{equation*}
where $m = v_2/2 = 1$. Thus $(X,Y)$ is a Nash equilibrium by Proposition~\ref{pr:5}.
If $b \in (0,1)$ then $(X,Y)$ is a Nash equilibrium by Proposition~\ref{pr:7:1}.
This completes proof of the right to left implication.

For the left to right implication, suppose that $(X,Y)$ is a Nash equilibrium of all-pay auction with $v_1 > v_2 = 2$.
Then, by Propositions~\ref{pr:5}, \ref{pr:6}, \ref{pr:7:1}, \ref{pr:7:2}, \ref{pr:8:1}, and~\ref{pr:8:2}, one of the following cases holds:
\begin{enumerate}
\item $\lfloor v_1/2 \rfloor = v_2/2$ or $\lceil v_1/2 \rceil = v_2/2+1$, $\Ex(X) = \Ex(Y) = v_2/2 = 1$,\label{p:2:1}
\item $\Ex(X) = 1$ and $\Ex(Y) = b$ and $b \in (0,1)$.\label{p:2:2}
\end{enumerate}
In case~\ref{p:2:1}, by Proposition~\ref{pr:5}, $X = \alpha\uo{1} + (1-\alpha)\ue{1}$ with $\alpha = 1$, so $X = \uo{1}$, and
$Y = \kappa \uo{1} + (1-\kappa) \ue{1} = (1-b)\uvec{0} + b(\kappa\uo{1} + (1-\kappa)\ue{1})$ with $b = 1$
and
\begin{equation*}
\kappa = \frac{2}{v_1}\left(2 - \frac{v_1}{2}\right) = \frac{4}{b v_1} - 1 \leq \frac{4}{b v_1} - \frac{2}{b} + 1
\end{equation*}
with $b = 1$. Hence $(X,Y)$ satisfies all the conditions stated in point~\eqref{th:int:2} of the theorem with $b = 1$.
In case~\ref{p:2:2}, by Proposition~\ref{pr:7:1}, $(X,Y)$ satisfies all the conditions stated in point~\eqref{th:int:2} of the theorem.
This completes proof of the left to right implication.

\textbf{Point~\eqref{th:int:3}.}\\
\noindent For point~\eqref{th:int:3} assume that $v_2/2 \in \mathbb{Z}$ and $v_1 > v_2 = 3$.

For the right to left implication take any strategy profile $(X,Y)$ with
$X = \uo{m}$ and $Y = (1 - b/m) \uvec{0} + (b/m)Z$,
where $m = v_2/2$, $b \in [v_2(v_2-2)/(2v_1),\min(m,v_2(v_2+2)/(2v_1))]$, and
\begin{equation*}
Z = \lambdao\uo{m} + \lambdae\ue{m} + \lambdaoup \uoup{m} + \sum_{j=1}^{m-1} \lambda_j \w{j}{m} 
\end{equation*}
with
\begin{equation*}
\lambdao,\lambdae,\lambdaoup,\lambda_1,\ldots,\lambda_{m-1} \geq 0 \textrm{ and } \lambdao + \lambdae + \lambdaoup + \sum_{j = 1}^{m-1} \lambda_j = 1,
\end{equation*}
\begin{equation*}
\frac{\lambdaoup}{m-1} - \frac{\lambdae}{m+1} = \frac{v_2^2}{2v_1b} - 1,
\end{equation*}
and
\begin{equation*}
\frac{\lambdaoup}{m-1} + \frac{1}{2m}\sum_{j = 1}^{m-1} \lambda_j \leq \frac{v_2}{2b}-1.
\end{equation*}

If $b = m$ then $X = \uo{m} = \alpha \uo{m} + (1-\alpha) \ue{m}$ with
\begin{equation*}
\alpha = 1 = \frac{2m}{v_2}\left(m+1-\frac{v_2}{2}\right)
\end{equation*}
and $Y = Z$.
Since 
\begin{equation*}
\frac{\lambdaoup}{m-1} + \frac{1}{2m}\sum_{j = 1}^{m-1} \lambda_j \leq \frac{v_2}{2b}-1 = 0
\end{equation*}
so $\lambdaoup,\lambda_1,\ldots,\lambda_{m-1} = 0$. In addition, since
\begin{equation*}
\frac{\lambdaoup}{m-1} - \frac{\lambdae}{m+1} = \frac{v_2^2}{2v_1b} - 1 = \frac{2m}{v_1} - 1,
\end{equation*}
and $\lambdaoup = 0$ so
\begin{equation*}
\lambdae = (m+1)\left(1 - \frac{2m}{v_1}\right) \textrm{ and } \lambdao = 1-\lambdae = \frac{2m}{v_1}\left(m+1-\frac{v_1}{2}\right).
\end{equation*}
Thus $(X,Y)$ is a Nash equilibrium by Proposition~\ref{pr:5}.
If $b \in [v_2(v_2-2)/(2v_1),v_2(v_2+2)/(2v_1)]$ then $(X,Y)$ is a Nash equilibrium by Proposition~\ref{pr:7:2}.
This completes proof of the right to left implication.

For the left to right implication, suppose that $(X,Y)$ is a Nash equilibrium of all-pay auction with $v_1 > v_2 = 3$.
Then, by Propositions~\ref{pr:5}, \ref{pr:6}, \ref{pr:7:1}, \ref{pr:7:2}, \ref{pr:8:1}, and~\ref{pr:8:2}, one of the following cases holds:
\begin{enumerate}
\item $\lfloor v_1/2 \rfloor = v_2/2$ or $\lceil v_1/2 \rceil = v_2/2+1$, $\Ex(X) = \Ex(Y) = v_2/2 = m$,\label{p:3:1}
\item $\Ex(X) = m = v_2/2$ and $\Ex(Y) = b$ with $m > b > 0$.\label{p:3:2}
\end{enumerate}
In case~\ref{p:3:1}, by Proposition~\ref{pr:5}, $X = \alpha\uo{m} + (1-\alpha)\ue{m}$ with $\alpha = 1$, so $X = \uo{m}$, and
$Y = \kappa \uo{m} + (1-\kappa) \ue{m} = (1-b/m)\uvec{0} + (b/m)(\kappa\uo{m} + (1-\kappa)\ue{m})$ with $b = m$
and
\begin{equation*}
\kappa = \frac{2m}{v_1}\left(m+1 - \frac{v_1}{2}\right).
\end{equation*}
Thus $Y = (1-b/m)\uvec{0} + (b/m) Z$ where
\begin{equation*}
Z = \lambdao\uo{m} + \lambdae\ue{m} + \lambdaoup \uoup{m} + \sum_{j=1}^{m-1} \lambda_j \w{j}{m} 
\end{equation*}
with $\lambdaoup,\lambda_1,\ldots,\lambda_{m-1} = 0$, $\lambdao = \kappa$, and $\lambdae = 1-\kappa$. Thus
\begin{equation*}
\frac{\lambdaoup}{m-1} - \frac{\lambdae}{m+1} = \frac{\kappa-1}{m+1} = \frac{2m}{v_1} - 1 = \frac{v_2^2}{2v_1b} - 1,
\end{equation*}
and
\begin{equation*}
\frac{\lambdaoup}{m-1} + \frac{1}{2m}\sum_{j = 1}^{m-1} \lambda_j = 0 = \frac{v_2}{2b}-1.
\end{equation*}
Hence $(X,Y)$ satisfies all the conditions stated in point~\eqref{th:int:3} of the theorem with $b = m$.
In case~\ref{p:3:2}, by Proposition~\ref{pr:7:2}, $(X,Y)$ satisfies all the conditions stated in point~\eqref{th:int:3} of the theorem.
This completes proof of the left to right implication.
\end{proof}

\begin{proof}[Proof of Theorem~\ref{th:nonint}]

For point~\eqref{th:nonint:1} assume that $v_2/2 \notin \mathbb{Z}$, $\lfloor v_1/2 \rfloor = \lfloor v_2/2 \rfloor$, and $v_2 > 2$.
The right to left implication follows immediately from Proposition~\ref{pr:5}. 
For the left to right implication, suppose that $(X,Y)$ is a Nash equilibrium of all-pay auction with $v_1$ and $v_2$ satisfying the conditions listed above.
Then, by Propositions~\ref{pr:5}, \ref{pr:6}, \ref{pr:7:1}, \ref{pr:7:2}, \ref{pr:8:1}, and~\ref{pr:8:2}, $\Ex(X) = \Ex(Y) = \lfloor v_2/2 \rfloor = m$
and, by Proposition~\ref{pr:5}, $(X,Y)$ satisfies all the conditions stated in point~\eqref{th:nonint:1} of the theorem.

For point~\eqref{th:nonint:3} assume that $v_2/2 \notin \mathbb{Z}$, $v_1/2 > \lfloor v_1/2 \rfloor + 1$, and $v_2 > 2$.
The right to left implication follows immediately from Proposition~\ref{pr:8:1}. 
For the left to right implication, suppose that $(X,Y)$ is a Nash equilibrium of all-pay auction with $v_1$ and $v_2$ satisfying the conditions listed above.
Then, by Propositions~\ref{pr:5}, \ref{pr:6}, \ref{pr:7:1}, \ref{pr:7:2}, \ref{pr:8:1}, and~\ref{pr:8:2}, $\Ex(X) = m+\alpha$ with $\alpha \in (0,1)$ and
$m = \lfloor v_2/2 \rfloor$, and $\Ex(Y) = b$ with $0 < b < m$ . Thus, by Proposition~\ref{pr:8:1}, $(X,Y)$ satisfies all the conditions stated in point~\eqref{th:nonint:3} of the theorem.

For point~\eqref{th:nonint:2} assume that $v_2/2 \notin \mathbb{Z}$, $v_1/2 = \lfloor v_1/2 \rfloor + 1$, and $v_2 > 2$.

For the right to left implication take any strategy profile $(X,Y)$ satisfying the conditions of point~\eqref{th:nonint:2} of the theorem.
If $(X,Y)$ satisfies the conditions with $\alpha > 0$ then $(X,Y)$ is a Nash equilibrium by Proposition~\ref{pr:8:2}.
If $(X,Y)$ satisfies the conditions with $\alpha = 0$ then $Y = \ue{m} = \kappa \uo{m} + (1-\kappa)\ue{m}$ with
\begin{equation*}
\kappa = \frac{2m}{v_1}\left(m+1-\frac{v_1}{2}\right) = 0
\end{equation*}
and
\begin{align*}
X & = \lambdao \uo{m} + \lambdae \ue{m} + \sum_{j = 1}^{m} \lambda_j \uo{m} + \sum_{j = 1}^{m} \kappa_j \ue{m},
\end{align*}
with $\lambdao,\lambdae,\lambda_1,\ldots,\lambda_m,\kappa_1,\ldots,\kappa_m \geq 0$ and
$\lambdao + \lambdae + \sum_{j = 1}^{m} \lambda_j + \sum_{j = 1}^{m} \kappa_j = 1$.
Let $\beta = \lambdao + \sum_{j = 1}^{m} \lambda_j$. Then
\begin{align*}
1 - \beta & = \lambdae + \sum_{i = 1}^m \kappa_i = \lambdae + \sum_{i = 1}^m \kappa_i \frac{1 - \alpha\delta}{1-\alpha}
          = \frac{\left\lceil \frac{v_2}{2} \right\rceil\left(\frac{v_2}{2} - \left\lfloor \frac{v_2}{2} \right\rfloor\right)}{\frac{v_2}{2}(1-\alpha)} - \frac{\alpha}{1-\alpha}\\
          & = (m+1)\left(1 - \frac{2m(m+1)}{v_2}\right) = 1 - \frac{2m}{v_2}\left(m+1 - \frac{v_2}{2}\right)
\end{align*}
Thus $(X,Y)$ is a Nash equilibrium by Proposition~\ref{pr:5}.
This completes proof of the right to left implication.

For the left to right implication, suppose that $(X,Y)$ is a Nash equilibrium of all-pay auction with $v_1$ and $v_2$ satisfying the conditions listed above.
Then, by Propositions~\ref{pr:5}, \ref{pr:6}, \ref{pr:7:1}, \ref{pr:7:2}, \ref{pr:8:1}, and~\ref{pr:8:2}, one of the following cases holds:
\begin{enumerate}
\item $\lfloor v_2/2 \rfloor = \lceil v_1/2 \rceil - 1$, $\Ex(X) = \Ex(Y) = \lfloor v_2/2 \rfloor = m$,\label{p:nonint:2:1}
\item $\Ex(X) = m + \alpha$ and $\Ex(Y) = m$ with $m = \lfloor v_2/2 \rfloor$.\label{p:nonint:2:2}
\end{enumerate}
In case~\ref{p:nonint:2:1}, by Proposition~\ref{pr:5}, $Y = \kappa\uo{m} + (1-\kappa)\ue{m}$ with $\kappa = 0$, so $Y = \ue{m}$, and
$X = \beta \uo{m} + (1-\beta) \ue{m}$ with
\begin{equation*}
\beta = \frac{2m}{v_2}\left(m+1 - \frac{v_2}{2}\right).
\end{equation*}
Thus
\begin{align*}
X & = \lambdao ((1-\alpha) \uo{m} + \alpha \uo{m+1}) + \lambdae ((1-\alpha) \ue{m} + \alpha \uo{m+1}) + {}\\
  & \qquad \sum_{j = 1}^{m} \lambda_j \left(\alpha\delta \vd{j}{m} + \left(1 - \alpha\delta \right) \uo{m}\right) + {}\\
  & \qquad \sum_{j = 1}^{m} \kappa_j \left(\alpha\delta \vd{j}{m} + \left(1 - \alpha\delta \right) \ue{m}\right),
\end{align*}
where
\begin{equation*}
\delta = \frac{2\left\lceil\frac{v_2}{2}\right\rceil - 1}{\left\lceil\frac{v_2}{2}\right\rceil},
\end{equation*}
with $\alpha = 0$, $\lambda_1 = \ldots = \lambda_m = \kappa_1 = \ldots = \kappa_m = 0$, $\lambdao = \beta$ and $\lambdae = 1 - \beta$.
Hence
\begin{align*}
\lambdae + \sum_{i = 1}^m \kappa_i \frac{1 - \alpha\delta}{1-\alpha} & = \lambdae = 1 - \frac{2m}{v_2}\left(m+1 - \frac{v_2}{2}\right)
= (m+1)\left(1 - \frac{2m(m+1)}{v_2}\right)\\
& = \frac{\left\lceil \frac{v_2}{2} \right\rceil\left(\frac{v_2}{2} - \left\lfloor \frac{v_2}{2} \right\rfloor\right)}{\frac{v_2}{2}(1-\alpha)} - \frac{\alpha}{1-\alpha}
\end{align*}
with $\alpha = 0$ so $(X,Y)$ satisfies all the conditions stated in point~\eqref{th:int:3} of the theorem with $\alpha = 0$.
In case~\ref{p:nonint:2:2}, by Proposition~\ref{pr:8:2}, $(X,Y)$ satisfies all the conditions stated in point~\eqref{th:nonint:2} of the theorem.
This completes proof of the left to right implication.
\end{proof}

\subsection*{Proof of Theorem~\ref{th:small}}

\begin{theorem}
\label{th:small}
Strategy profile $(X,Y)$ is a Nash equilibrium of all-pay auction with players valuations $v_1 \geq v_2 > 0$ and $v_2 < 2$ if and only if
$Y = \uvec{0}$ and
\begin{itemize}
\item $\frac{v_1}{2} > 1$ and $X = \uvec{1}$,
\item $\frac{v_1}{2} = 1$ and $X = (1 - \alpha) \uvec{0} + \alpha \uvec{1}$ where $\alpha \in [0,1]$,
\item $\frac{v_1}{2} < 1$ and $X = \uvec{0}$.
\end{itemize}

Equilibrium payoffs of the players are
\begin{itemize}
\item if $\frac{v_1}{2} > 1$ then $P^1(X,Y) = v_1 - 1$,
\item if $\frac{v_1}{2} = 1$ then $P^1(X,Y) = (1+\alpha)\frac{v_1}{2} - \alpha$,
\item if $\frac{v_1}{2} < 1$ then $P^1(X,Y) = 0$,
\end{itemize}
and $P^2(X,Y) = 0$.
\end{theorem}

\begin{proof}
For the left to right implication, let $(X,Y)$ be a Nash equilibrium of all-pay auction with players valuations $v_1 \geq v_2 > 0$ and $v_2 < 2$.
By Proposition~\ref{pr:allpaylotto}, $(X,Y)$ is a Nash equilibrium of General Lotto game $\Gammait(\Ex(X),\Ex(Y))$. Hence, by Propositions~\ref{pr:5},
\ref{pr:6}, \ref{pr:7:1}, \ref{pr:7:2}, \ref{pr:8:1}, and \ref{pr:8:2}, either $\Ex(Y) = 0$ or $\Ex(X) = 0$.
Suppose that $\Ex(Y) = 0$. Then $Y = \uvec{0}$ and, by~\eqref{eq:huvec0} and~\eqref{eq:hprop:1},
\begin{equation*}
H(X,\uvec{0}) = 1 - \Pr(X = 0).
\end{equation*}
Two cases of values of $\Ex(X)$ are possible: $\Ex(X) \in [0,1)$ and $\Ex(X) \geq 1$.
Suppose that $\Ex(X) = m \geq 1$. By~\eqref{eq:pandh} payoff to player $1$ from strategy profile $(X,Y)$ is
\begin{equation*}
P^1(X,Y) = \frac{v_1}{2}\left(H(X,Y) - \frac{2\Ex(X)}{v_1} + 1\right).
\end{equation*}
Since $m \geq 1$ so there exists $X$ with $\Pr(X = 0)$ and $\Ex(X) = m$ (e.g. a convex combination of $\uvec{\lfloor m \rfloor}$ and $\uvec{\lceil m \rceil}$)
so
\begin{equation*}
P^1(X,Y) \geq \frac{v_1}{2}\left(1 - \frac{2\Ex(X)}{v_1} + 1\right) = v_1 - \Ex(X).
\end{equation*}
Thus, in the case of $\Ex(X) = m \geq 1$, $P(X,Y)$ is maximised when $\Ex(X) = 1$ and $\Pr(X = 0) = 0$, in which case $X = \uvec{1}$.

Suppose that $\Ex(X) \in [0,1)$.
Then
\begin{equation*}
H(X,\uvec{0}) = 1 - \Pr(X = 0) = \sum_{j \geq 1} \Pr(X \geq j) - \sum_{j \geq 2} \Pr(X \geq j) = \Ex(X) - \sum_{j \geq 2} \Pr(X \geq j) \leq \Ex(X)
\end{equation*}
with equality when $\Pr(X \geq 2) = 0$. Thus the equality is attained by strategy $(1-\alpha)\uvec{0} + \alpha\uvec{1}$, where $\alpha = \Ex(X)$.
By~\eqref{eq:pandh} payoff to player $1$ from strategy profile $(X,Y)$ is
\begin{equation*}
P^1(X,Y) = \frac{v_1}{2}\left(H(X,Y) - \frac{2\Ex(X)}{v_1} + 1\right) \geq \frac{v_1}{2}\left(\Ex(X) - \frac{2\Ex(X)}{v_1} + 1\right)
= \Ex(X)\left(\frac{v_1}{2} - 1 \right) + \frac{v_1}{2}.
\end{equation*}

Since $(X,Y)$ is a Nash equilibrium so $\Ex(X) = 0$ and $X = \uvec{0}$ if $v_1/2 < 1$, 
$\Ex(X) = \alpha \in [0,1]$ and $X = (1-\alpha)\uvec{0} + \alpha \uvec{1}$, if $v_1/2 = 1$, and $\Ex(X) = 1$ and $X = \uvec{1}$ if $v_1/2 > 1$.

It is elementary to verify that the strategy profiles stated in the theorem are Nash equilibria for the corresponding values of $v_1$ and $v_2$.
\end{proof}

\end{document}